\documentclass[aps,pra,twocolumn,superscriptaddress,floatfix,nofootinbib,showpacs,longbibliography]{revtex4-1}

\usepackage{xcolor}
\usepackage[utf8]{inputenc}  
\usepackage[T1]{fontenc}     
\usepackage[british]{babel}  
\usepackage[sc,osf]{mathpazo}
\usepackage[scaled=0.86]{berasans}  
\usepackage[colorlinks=true, citecolor=blue, urlcolor=blue]{hyperref}  
\usepackage[babel]{microtype}  
\usepackage{amsmath,amssymb,amsthm,bm,amsfonts,mathrsfs,bbm} 

\usepackage{xspace}  
\usepackage{pgf,tikz}
\usepackage{xcolor}
\usepackage{multirow}
\usepackage{array}
\usepackage{bigstrut}
\usepackage{braket}
\usepackage{color}
\usepackage{natbib}
\usepackage{multirow}
\usepackage{mathtools}
\usepackage{float}
\usepackage[caption = false]{subfig}
\usepackage{xcolor,colortbl}
\usepackage{color}

\newcommand{\be}{\begin{equation}}
\newcommand{\ee}{\end{equation}}
\newcommand{\ba}{\begin{eqnarray}}
\newcommand{\ea}{\end{eqnarray}}

\newtheorem{theorem}{Theorem}

\def\>{\rangle}
\def\<{\langle}

\newcommand{\orcid}[1]{\href{https://orcid.org/#1}{\includegraphics[width=10pt]{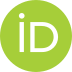}}}
\begin{document}
\title{Classification and Quantification of Entanglement Through Wedge Product and Geometry}

\author{Soumik Mahanti\orcid{0000-0002-0380-0324}}
\email{soumikmh1998@gmail.com}
\affiliation{Department of Physical Sciences, IISER Kolkata, Mohanpur 741246, West Bengal, India}

\author{Sagnik Dutta\orcid{0000-0003-2083-9922}}
\email{sagnikdutta17@gmail.com}
\affiliation{Department of Physical Sciences, IISER Kolkata, Mohanpur 741246, West Bengal, India}

\author{Prasanta K. Panigrahi\orcid{0000-0001-5812-0353}}
\email{panigrahi.iiser@gmail.com}
\affiliation{Department of Physical Sciences, IISER Kolkata, Mohanpur 741246, West Bengal, India}

\begin{abstract}
 Wedge product of post-measurement vectors leading to an `area' measure of the parallelogram has been shown to give the generalized I-concurrence measure of entanglement.  Extending the wedge product formalism to multi qudit systems, we have presented a modified faithful entanglement measure, incorporating the higher dimensional volume and the area elements of the parallelepiped formed by the post-measurement vectors. The measure fine grains the entanglement monotone, wherein different entangled classes manifest with different geometries. We have presented a complete analysis for the bipartite qutrit case considering all possible geometric structures. Three entanglement classes can be identified with different geometries of post-measurement vectors, namely three planar vectors, three mutually orthogonal vectors, and three vectors that are neither planar and not all of them are mutually orthogonal. It is further demonstrated that the geometric condition of area and volume maximization naturally leads to the maximization of entanglement. The wedge product approach uncovers an inherent geometry of entanglement and is found to be very useful for characterization and quantification of entanglement in higher dimensional systems.

\end{abstract}

\maketitle

\section{Introduction}
Quantum entanglement has fundamental importance in both applications of quantum information processing and the foundational understanding of quantum theory \cite{wilde2013quantum}. It is the most useful resource for various quantum communication protocols like quantum teleportation \cite{Bennet_1993,D_Saha,Sagnik,Joo,Hu}, dense coding \cite{Bennet_1992,Ujjwal1,Ujjwal2,Pati,Murali,Guo}, and secret sharing  \cite{Hillery_1999,Gottesman,Singh}, which gives the quantum advantage over the classical communication protocols. Therefore, it is of practical interest to quantify this resource to estimate the efficiency of such protocols. Various aspects of quantum entanglement have been studied extensively \cite{Horodecki_2009}, however the classification and quantification of entangled states is still not well understood in higher dimensional systems \cite{Bastin,Ribeiro,LiLi,Gilad,Ghahi,Home}. Entangled qudits have advantage over entangled two level systems in several communication and cryptographic protocols \cite{Peres,Collins,qutrit}, exhibiting stronger nonlocality in maximally entangled states \cite{Kaszlikowski}. In short, the study of characterization and quantification of entanglement is a very rich area of research even today \cite{Gharahi2,Eberly_2021}, particularly from a geometric perspective, which can provide insight into the distributive nature of entanglement.

Concurrence, introduced by Hill and Wootters, is a faithful measure of entanglement that provides the necessary and sufficient conditions of separability for a pair of qubits \cite{Hill_1997}. Rungta et al. extended the concurrence to bipartite pure states in $d_1\otimes d_2$ systems \cite{Rungta_PRA_2001}. Various geometry based approaches are also there to quantify entanglement \cite{Horogeo,Goldbart,Bhaskar,Abhinash,Shreya,Gugeo,Kh}. Bhaskara and Panigrahi used the wedge product framework and Lagrange-Brahmagupta identity to provide for an entanglement measure based on geometry for pure states \cite{Bhaskar}. The Bhaskara-Panigrahi measure \cite{Bhaskar,Abhinash,Shreya,Qureshi} of entanglement across any bipartition of a multiparty state is the amount of total area formed by the post-measurement vectors of a subsystem.p. If all post-measurement vectors in a subsystem are parallel, then the state is separable across that bipartition. A pure composite state is separable, if it is separable across all possible bipartitions. This geometrical condition greatly simplifies the separability criteria in multi qudit systems. It has been further shown that the measure yields the generalized I-concurrence in two qudit scenario, making it a good measure of entanglement.

In this work, we present a modified faithful entanglement measure based on the wedge product formalism. Higher dimensional volume and area, denoted by multiple wedge products, is incorporated into our new entanglement measure. In any bipartition, wedge products between all post-measurement vectors construct a hypervolume of the parallelepiped formed by the vectors. Similarly, wedge product between multiple but not all post-measurement vectors render the area of the hypersurface formed by those vectors. Incorporating these terms into the measure reveals insightful geometrical structures of the entangled states. We present a complete analysis of all possible geometries of pure bipartite qutrit states using our new entanglement measure. The general structure for entangled state in two qutrit systems are parallelepiped formed by the three post-measurement vectors in $\mathcal{C}^3$. Two qutrit systems have three entangled classes \cite{PanLu}, which are shown to construct three distinct types of geometric structures that are not convertible from one another under local unitary operations. These classes have the geometries of three planar vectors, three mutually orthogonal vectors, and three vectors that are neither planar nor all mutually orthogonal. We further demonstrate that the maximization of volume and area of the parallelepiped for a particular type yields the maximally entangled state of that type and those conditions are easily deduced from the geomtery. 

The paper is organized as follows. We introduce our entanglement measure for two qutrit systems in section \ref{section_2}, and show why it is a faithful measure of entanglement. In section \ref{section3}, we discuss about how entanglement is classified and what are the entangled classes for two qutrits. In section \ref{section4}, we describe the geometric structures and its connection to the entangled classes. The entanglement maximization conditions from geometry are also deduced. Generalized geometric measure for multi qudit pure state is presented in section \ref{section 5}.  

\section{Two qutrit entanglement measure}
\label{section_2}
A pure two qutrit state can be written in the computational basis of Alice and Bob as 
\begin{equation}\label{eq1}
\begin{split}
    \ket{\psi} = & a\ket{00}+b\ket{01}+c\ket{02}+p\ket{10}+q\ket{11}\\
    &+r\ket{12}+x\ket{20}+y\ket{21}+z\ket{22}
\end{split}
\end{equation}
with the normalization condition $|a|^2+|b|^2+|c|^2+|p|^2+|q|^2+|r|^2+|x|^2+|y|^2+|z|^2 = 1$.

The state can be rewritten as a convex sum of product vectors, where Alice's subsystem is in the computational basis.
\begin{equation}\label{eq2}
\ket{\psi} = \ket{0}_A\ket{\eta_0}_B+\ket{1}_A\ket{\eta_1}_B+\ket{2}_A\ket{\eta_2}_B    
\end{equation}
Since, if Alice measures her subsystem in the computational basis, the state of Bob will be $\frac{\ket{\eta_i}_B}{\braket{\eta_i|\eta_i}_B}$ corresponding to Alice's measurement result $i$, we call the vectors $\ket{\eta_i}_B$ as post-measurement vectors. Here,
\begin{equation}\label{eq3}
\begin{split}
\ket{\eta_0}_B = a\ket{0}+b\ket{1}+c\ket{2}\\
\ket{\eta_1}_B = p\ket{0}+q\ket{1}+r\ket{2}\\
\ket{\eta_2}_B = x\ket{0}+y\ket{1}+z\ket{2}
\end{split}
\end{equation}
$\ket{\eta_i}_B$ are the unnormalized post-measurement vectors for Bob's subsystem. From Eq. \ref{eq2}, it can be shown that the state takes a product form only if for some complex number $m$ and $n$, the following relation holds $\ket{\eta_1}_B = m\ket{\eta_0}_B, \text{and} \ket{\eta_2}_B = n\ket{\eta_0}_B$. Then, the state can be written as $\ket{\psi}=(\ket{0}_A+m\ket{1}_A+n\ket{2}_A)\otimes \ket{\eta_0}_B$ or in other similar product forms.

So, the state is separable if and only if the post measurement vectors $\ket{\eta_0}_B, \ket{\eta_1}_B \text{and} \ket{\eta_2}_B$ are parallel to each other. Using this fundamental connection of separability with the parallelism of complex vectors of a subsystem, Bhaskara and Panigrahi gave an entanglement measure for pure states that yields the generalized I-concurrence.
In this work, a modified faithful entanglement measure is presented incorporating multiple wedge product between the post-measurement vectors.

Our modified entanglement measure based on geometry of post-measurement vectors, called $E_G$ for two qutrit systems is given by:
\begin{equation}\label{eq6}
\begin{split}
     E_G &= 9 |\ket{\eta_0}\wedge\ket{\eta_1}\wedge\ket{\eta_2}|^2+2(|\ket{\eta_0}\wedge\ket{\eta_1}|^2\\
            &+|\ket{\eta_1}\wedge\ket{\eta_2}|^2+|\ket{\eta_2}\wedge\ket{\eta_0}|^2)
\end{split}
\end{equation}
In terms of the continuous parameters $a,b,\dots,z$, the measure can be written as 
\begin{equation}\label{eq7}
    E_G = 9 |\det A|^2 + 2 \left[ \begin{Vmatrix}
    br-cq\\
    cp-ar\\
    aq-bp
    \end{Vmatrix}^2 + \begin{Vmatrix}
    bz-cy\\
    cx-az\\
    ay-bx
    \end{Vmatrix}^2+ \begin{Vmatrix}
    qz-ry\\
    xr-pz\\
    py-qx
    \end{Vmatrix}^2 \right]
\end{equation}
where $A = \left(\begin{smallmatrix}
a &b&c\\p &q&r\\x &y&z \end{smallmatrix}\right)$, and the $\lVert .\rVert$ sign is used to represent norm of a column vector. Expanding by terms, the full expression is the following 

\begin{equation}\label{eq8}
\begin{split}
E_G=  &9|a(qz-ry)-b(pz-rx)+c(py-qx)|^2\\&+2(|br-cq|^2+|cp-ar|^2+|aq-bp|^2)\\&+2(|bz-cy|^2+|cx-az|^2+|ay-bx|^2)\\&+2(|qz-ry|^2+|xr-pz|^2+|py-qx|^2)
\end{split}
\end{equation}

Geometrically, $\ket{\eta_0}$, $\ket{\eta_1}$, and $\ket{\eta_2}$ are three vectors in $\mathcal{C}^3$ Hilbert space. The first term of the Eq. \ref{eq6} geometrically represents the modulus square of the volume of the parallelopiped formed by those vectors. The rest three terms are square of all six surface areas of the 3 dimensional parallelepiped. The factor of 9 has been put before the volume element to restrict the measure $E_G$ to take value between 0 and 1. It is 1 for the maximally entangled states of the form $\frac{1}{\sqrt{3}}(\ket{00}+\ket{11}+\ket{22})$, and 0 for separable states. 
Any measure of entanglement is considered as a `good' measurement if it follows the following three conditions.
\begin{itemize}
   \item The measure must give a value 0 for any separable state.
   \item Value of the measure should be local unitary invariant. 
   \item Average entanglement value after LOCC must not increase.
\end{itemize}

We have already shown that $E_G$ satisfies the first condition since parallelism of all vectors make the wedge product zero. Under local unitary operation on the second sub-system, the length and the angle between the post-measurement vectors of that subsystem are unchanged. Hence, the wedge products, and consequently, $E_G$  remains invariant. Again, it is seen that the entanglement measure $E_G$ is party symmetric, which implies that its value is independent of whether one takes the post-measurement vectors of subsystem 1 or 2. Hence, local unitary operation on the first subsystem will also leave the value of $E_G$ unchanged. Therefore, the second condition is also satisfied. We only consider single copy pure state scenario, for which local unitary invariance is same as LOCC invariance \cite{Horodecki_2009}. Hence, the entanglegemnt measure $E_G$ satisfies all the criteria for a faithful measure of entanglement. 
\section{Entanglement classes for bipartite qutrit systems}\label{section3}
 As entanglement appeared to be a very useful resource for quantum communication, the important set of operations emerged as local operations. The idea behind is that distant parties can influence their sub-system only. Hence, local operations and classical communication between distant parties are the only available means for the manipulation of entangled states for communication purposes. Nielsen provided with a significant result about transformation criteria between two bipartite pure states under LOCC \cite{Nielsen}. But exact transformation under LOCC lacks continuity, and the Nielsen condition is not applicable for multiparty scenario. The necessary and sufficient condition for LOCC transformation between two pure states is if they can be transformed under local unitary operation. Dur et al. proposed another classification encapsulating qualitative features of entanglement in terms of stochastic local operation and classical communication \cite{Dur}. Two pure states are said to be equivalent if they can be transformed into each other by LOCC with some nonzero probability. Mathematically, if $\ket{\psi} = A_1\otimes\dots\otimes A_n\ket{\phi}$, with $A_i$ being reversible operators, then $\ket{\psi}$ and $\ket{\phi}$ are equivalent. Entangled states that cannot be transformed into each other by SLOCC, belong to different entangled classes. Previous results in the search of entangled classes have shown that there are three inequivalent entangled classes for three qubit pure states \cite{PanLu}. Bipartite pure states in $d\otimes d$ dimension have $d$ entangled classes of states \cite{Horodecki_2009}. Several works have been done regarding classification and characterization of entangled states in pure and mixed scenario \cite{Acinn,Miyake,Miyakee,Verstraete_four_qubits,Guhne,nandi2022wigner}. We present the classification of bipartite qutrit entangled states based on our geometric measure of entanglement. Our analysis considers pure states with a single copy available scenario. 
\section{Geometrical interpretation of entanglement classes}\label{section4}
From Eq. \ref{eq3}, it is seen that there are three post-measurement vectors, belonging to $\mathcal{C}^3$. For geometrical representation in real space, $\ket{0}$ is represented in the X-axis, $\ket{1}$ in the Y-axis, and $\ket{2}$ in the Z-axis. Post-measurement vector of the form $a\ket{0}+b\ket{1}+c\ket{2}$ is represented as a point vector in $\mathcal{R}^3$. All the mathematical results presented here are valid for any complex number, but this 3-dimensional real representation captures all the intricate geometries of the entangled states in bipartite qutrit systems. Below, three theorems are presented to show how the entanglement classes are related to the geometry. 
\begin{theorem}\label{Th1}
Three planar vectors remain in the same plane after local unitary operation in any of the subsystem.
\end{theorem}
\begin{proof}
: From Eq. \ref{eq2}, an arbitrary two qutrit state can be written as $\ket{\psi} = \ket{0}_A\ket{\eta_0}_B+\ket{1}_A\ket{\eta_1}_B+\ket{2}_A\ket{\eta_2}_B$. Applying any local unitary in the second subsystem does not change the angle between the vectors $\ket{\eta_0}_B$, $\ket{\eta_1}_B$, and $\ket{\eta_2}_B$. Hence, only the orientation of the parallelepiped can change but not the shape. Now, let Alice apply an arbitrary local unitary operation $U$ in her subsystem. 
$U=\begin{pmatrix}x_{00}&x_{01}&x_{02}\\x_{10}&x_{11}&x_{12}\\x_{20}&x_{21}&x_{22}
\end{pmatrix}$.\\ \vspace{0.15cm}
$U\ket{0} = x_{00}\ket{0}+x_{10}\ket{1}+x_{20}\ket{2}, \hspace{0.15cm}
    U\ket{1} = x_{01}\ket{0}+x_{11}\ket{1}+x_{21}\ket{2}, \hspace{0.15cm} \text{and}\hspace{0.15cm}
    U\ket{0} = x_{02}\ket{0}+x_{12}\ket{1}+x_{22}\ket{2}$. Therefore, $(U\otimes \mathbb{I})\ket{\psi} =
\ket{0}_A\ket{\chi_0}_B+\ket{1}_A\ket{\chi_1}_B+\ket{2}_A\ket{\chi_2}_B$. 
Where,
\begin{equation}\label{eq9}
\begin{split}
    \ket{\chi_0}_B=x_{00}\ket{\eta_{0}}+x_{01}\ket{\eta_{1}}+x_{02}\ket{\eta_{2}}\\
    \ket{\chi_1}_B=x_{10}\ket{\eta_{0}}+x_{11}\ket{\eta_{1}}+x_{12}\ket{\eta_{2}}\\
    \ket{\chi_2}_B=x_{20}\ket{\eta_{0}}+x_{21}\ket{\eta_{1}}+x_{22}\ket{\eta_{2}}
\end{split}
\end{equation}
If the vectors lie in a plane, then without loss of generality, $\ket{\eta_2}=m\ket{\eta_0}+n\ket{\eta_1}$. Therefore,
\begin{align*}
\ket{\chi_0}_B=(x_{00}+mx_{20})\ket{\eta_0}_B+(x_{01}+nx_{02})\ket{\eta_1}_B\\ \ket{\chi_1}_B=(x_{10}+mx_{12})\ket{\eta_0}_B+(x_{11}+nx_{21})\ket{\eta_1}_B\\ \ket{\chi_2}_B=(x_{20}+mx_{22})\ket{\eta_0}_B+(x_{12}+nx_{22})\ket{\eta_1}_B
\end{align*}
So, the post measurement vectors $\ket{\chi_0}, \ket{\chi_1}, \text{and} \ket{\chi_2}$ lie in the same plane as do $\ket{\eta_0}, \ket{\eta_1}, \text{and} \ket{\eta_2}$. Therefore, the planar structure is retained after applying local unitary to either of the subsystem.
\end{proof}
\begin{theorem}\label{Th2}
Three orthogonal vectors remain orthogonal after local unitary operation in any of the subsystem.
\end{theorem}
\begin{proof}
: As stated in the proof of Theorem \ref{Th1}, only local unitary operation on the first subsystem can change the shape of the parallelepiped. If a local unitary operator U is applied on the first subsystem, the state transforms according to Eq. \ref{eq9}.
\begin{equation}
\begin{split}
(U\otimes \mathbb{I})\ket{\psi}=&\ket{0}_A(x_{00}\ket{\eta_{0}}+x_{01}\ket{\eta_{1}}+x_{02}\ket{\eta_{2}})_B\\
+&\ket{1}_A(x_{10}\ket{\eta_{0}}+x_{11}\ket{\eta_{1}}+x_{12}\ket{\eta_{2}})_B\\
+&\ket{2}_A(x_{20}\ket{\eta_{0}}+x_{21}\ket{\eta_{1}}+x_{22}\ket{\eta_{2}})_B     
\end{split}
\label{unitary}
\end{equation}
 Given the condition, $\braket{\eta_i|\eta_j}_B= 0$, for $i\neq j$, we need to show that $\braket{\chi_i|\chi_j}_B= 0$, for $i \neq j$.\\
$UU^\dagger = \begin{pmatrix}x_{00}&x_{01}&x_{02}\\x_{10}&x_{11}&x_{12}\\x_{20}&x_{21}&x_{22}\end{pmatrix}\begin{pmatrix}x_{00}^*&x_{10}^*&x_{20}^*\\x_{01}^*&x_{11}^*&x_{21}^*\\x_{02}^*&x_{12}^*&x_{22}^*\end{pmatrix}= \mathbb{I}$. Therefore, the off-diagonal terms are zero and we get the conditions:
\begin{equation}
    \begin{split}\label{Cond}
        x_{00}^*x_{01}+x_{10}^*x_{11}+x_{20}^*x_{21}=0=\braket{\chi_1|\chi_0}_B\\
        x_{00}^*x_{02}+x_{10}^*x_{12}+x_{20}^*x_{22}=0=\braket{\chi_0|\chi_2}_B\\
        x_{01}^*x_{02}+x_{11}^*x_{12}+x_{21}^*x_{22}=0=\braket{\chi_1|\chi_2}_B
    \end{split}
\end{equation}
Hence it is proved that the vectors remain orthogonal after applying local unitary to either of the subsystems.
\end{proof}

\begin{theorem}
Three mutually non-orthogonal vectors are local unitarily equivalent to one pair or two pair of orthogonal vectors.
\end{theorem}
\begin{proof}
: Like in the previous cases, under the local unitary U on the first subsystem, the state takes the form of Eq. \ref{unitary}. Let $\braket{\eta_0|\eta_1}_B= l,\braket{\eta_1|\eta_2}_B= m,\braket{\eta_2|\eta_0}_B= n$. This gives:
\begin{equation}\label{N-O}
\begin{split}
\braket{\chi_0|\chi_1}_B=&x_{00}^*(lx_{11}+n^*x_{12})+x_{01}^*(l^*x_{10}+mx_{12})\\+&x_{02}^*(nx_{10}+m^*x_{11})\\
\braket{\chi_0|\chi_2}_B=&x_{00}^*(lx_{21}+n^*x_{22})+x_{01}^*(l^*x_{20}+mx_{22})\\+&x_{02}^*(nx_{20}+m^*x_{21})\\
\braket{\chi_1|\chi_2}_B=&x_{10}^*(lx_{21}+n^*x_{22})+x_{11}^*(l^*x_{20}+mx_{22})\\+&x_{12}^*(nx_{20}+m^*x_{21})
\end{split}
\end{equation}
Now, let $l = m = 0$, and $n\neq0$. Then from Eq. \ref{N-O}, $\braket{\chi_0|\chi_1}_B=n^*x_{00}^*x_{12}+nx_{02}^*x_{10}$, $\braket{\chi_0|\chi_2}_B=n^*x_{00}^*x_{22}+nx_{02}^*x_{20}$, and $\braket{\chi_1|\chi_2}_B=n^*x_{10}^*x_{22}+nx_{12}^*x_{20}$. In general, all three inner products are nonzero unless $n=0$. Hence, two pairs of orthogonal vectors can be local unitarily transformed into three mutually non-orthogonal vectors, and vice versa. Similarly, taking one orthogonal pair of vectors, it can be shown that local unitary operation can transform it to three mutually non-orthogonal vectors. Therefore, the proof is complete.
\end{proof}

It is easily seen that the new measure of entanglement based on wedge product of complex vectors construct different geometry to different entangled states. Under local unitary, the geometry of the post-measurement vectors can not be transformed arbitrarily. For pure states, local unitary transformation is equivalent to LOCC transformation. Since any of the three geometries have zero probability to be transformed into another under LU, they are not connected by stochastic LOCC. Hence, they constitute different entangled classes. These theorems encompass all possible geometry that can arise from the 3d geometry of the parallelepiped formed by the post-measurement vectors. Our results show that how simple 3d geometry can distinguish between the different classes of two qutrit entangled states.

The two qutrit pure states have nine basis elements. Taking linear combination of two or more terms, we present a case by case analysis for all possible states and their geometric shapes. We carry on the analysis in similar spirit as done by Pan et al. \cite{PanLu}. The geometrical shapes, entanglement maximization conditions and the entangled types are presented in table formats.

\subsection{Two term states}
The nine terms in the computational basis of two qutrits can be listed as $\left(\begin{smallmatrix}\ket{00} &\ket{01} &\ket{02} \\
\ket{10} &\ket{11} &\ket{12} \\
\ket{20} &\ket{21} &\ket{22} \end{smallmatrix}\right)$. \\Picking combination of any two terms from the same row or same column will not be entangled. Entangled states consisting any two terms have the same geometric shape. For the example state in Table \ref{two}, the post-measurement vectors are along the X and Y axis, and the shape is of a rectangle with the side lengths $|a|\text{, and }|b|$ respectively. The area is maximum, if each side is of equal length to make it a square. The entanglement maximization condition is thus $|a|=|q|$

\begin{widetext}

\begin{table}[htbp]
    \centering
    \begin{tabular}{|c|c|c|c|c|}
    \hline
    \textbf{Example}&\textbf{Max. Ent.}&\textbf{Figure}&\textbf{Geometry}&\textbf{Ent.}\\
    \textbf{State ($\ket{\psi}$)}&\textbf{Value}&\textbf{(Example)}&\textbf{Structure}&\textbf{Class}\\
    \hline
    \begin{minipage}{2.5 cm}
    $a\ket{00}+q\ket{11}$ \vspace{0.1 cm}\\
    $~~=\ket{0}_A(a\ket{0}_B)$\\
    $~~~~+\ket{1}_A(q\ket{1}_B)$
    \end{minipage}& 
    \begin{minipage}{2 cm}
    $E_{Gm}=\frac{1}{2}$ \vspace{0.2cm}\\
    \textbf{Condition:}\\
    $|a|=|q|$
    \end{minipage} &
    \begin{minipage}{5 cm}
    \vspace{0.1 cm}
      \includegraphics[width=1\linewidth,trim={2.5cm 1.5cm 2.5cm 2.5cm},clip]{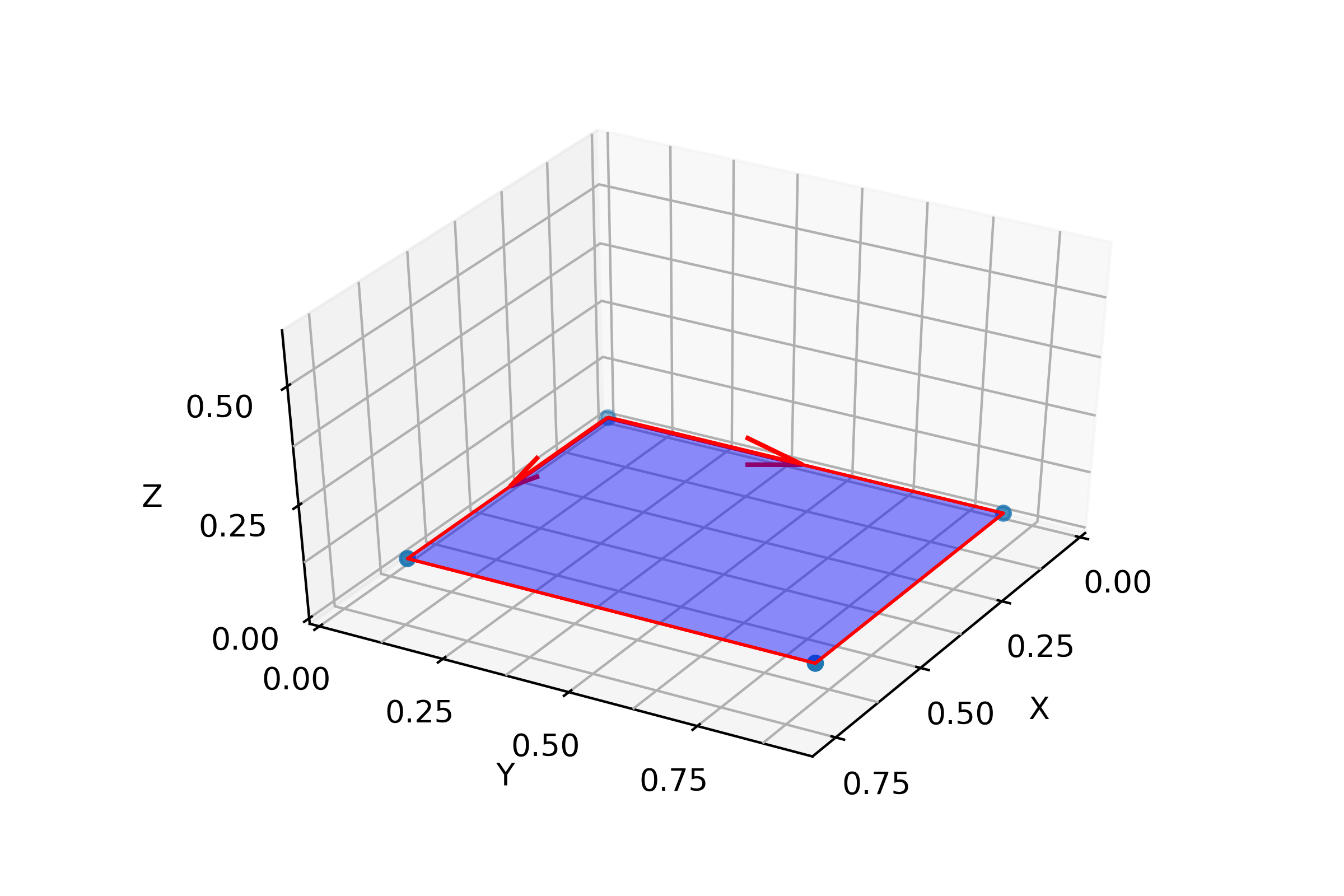}
    \end{minipage}&
    \begin{minipage}{2.5 cm}
     \textbf{Planer} \vspace{0.1 cm}\\
     No. of \\
     Orthogonal
     pair(Op):\\
     $Op=1$
    \end{minipage}
    &\textbf{Type I}\\
    \hline
    \end{tabular}
    \caption{Geometric structure and entanglement type of two term states}
    \label{two}
\end{table}
\end{widetext}
 
 \subsection{Three terms}
Picking three terms from the basis can be done in three inequivalent ways. Their geometric shape and entanglement types are shown in Table \ref{three}. From the geometrical shape, the maximum entanglement value found for planar shape is 0.5, but is 1 for non-planar structure. These two structures give two different types of entanglement classes. The planar structure corresponds to Type I entanglement class and the non-planar structure with three mutually orthogonal pair is of Type II entanglement class.

\begin{widetext}

\begin{table}[htbp]
    \centering
    \begin{tabular}{|c|c|c|c|c|}
    \hline
    \textbf{Example}&\textbf{Max. Ent.}&\textbf{Figure}&\textbf{Geometry}&\textbf{Ent.}\\
    \textbf{State ($\ket{\psi}$)}&\textbf{Value}&\textbf{(Example)}&\textbf{Structure}&\textbf{Class}\\
    \hline
    
    \begin{minipage}{4 cm}
    $a\ket{00}+q\ket{11}+z\ket{22}$ \vspace{0.1 cm}\\
    $~~=\ket{0}_A(a\ket{0}_B)+\ket{1}_A$\\
    $(q\ket{1}_B)+\ket{2}_A(q\ket{2}_B)$
    \end{minipage}& 
    \begin{minipage}{2.5 cm}
    $E_{Gm}=1$ \vspace{0.2cm}\\
    \textbf{Condition:}\\
    $|a|=|q|=|z|=\frac{1}{\sqrt{3}}$\\
    \end{minipage} &
    \begin{minipage}{5 cm}
    \vspace{0.1 cm}
      \includegraphics[width=1\linewidth,trim={2.5cm 2.5cm 2.5cm 2.5cm},clip]{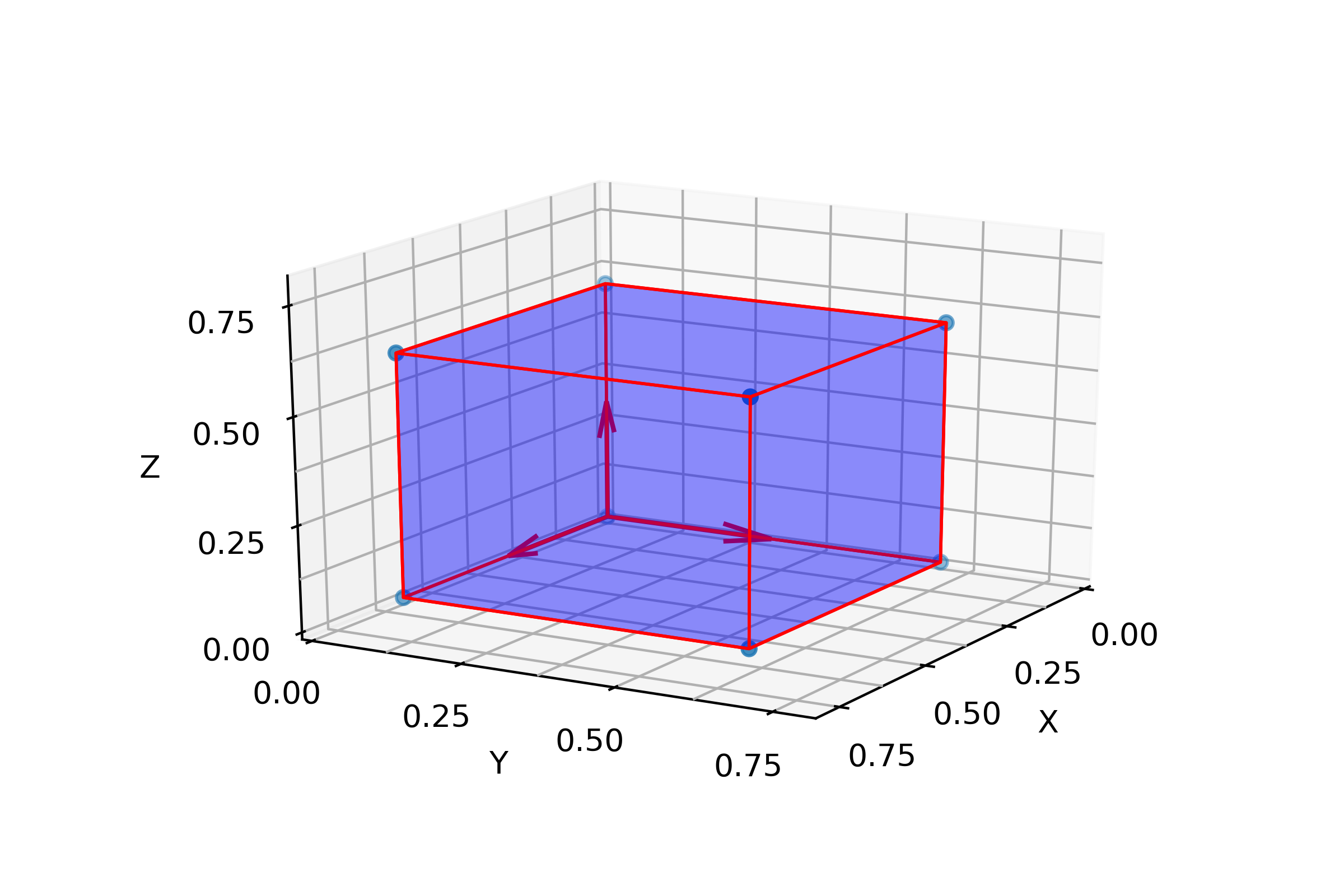}
    \end{minipage}&
    \begin{minipage}{2.5 cm}
    \textbf{Non-Planer} \vspace{0.1 cm}\\
     No. of
     Orthogonal
     pair(Op):\\
     $Op=3$
    \end{minipage}
    &\textbf{Type II}\\
    \hline
    
    \begin{minipage}{4 cm}
    $a\ket{00}+q\ket{01}+z\ket{10}$ \vspace{0.1 cm}\\
    $~~=\ket{0}_A(a\ket{0}_B+b\ket{1}_B)$\\
    $~~~~~~~~~+\ket{1}_A(p\ket{0}_B)$
    \end{minipage}& 
    \begin{minipage}{2.5 cm}
    $E_{Gm}<\frac{1}{2}$ \vspace{0.2cm}\\
    \textbf{Remarks:}\\
    Maximizes as $a \to 0$;\\
    $|b|^2,|p|^2$ $\rightarrow \frac{1}{2}$
    \end{minipage} &
    \begin{minipage}{5 cm}
    \vspace{0.1 cm}
      \includegraphics[width=1\linewidth,trim={2.5cm 2.5cm 2.5cm 2.5cm},clip]{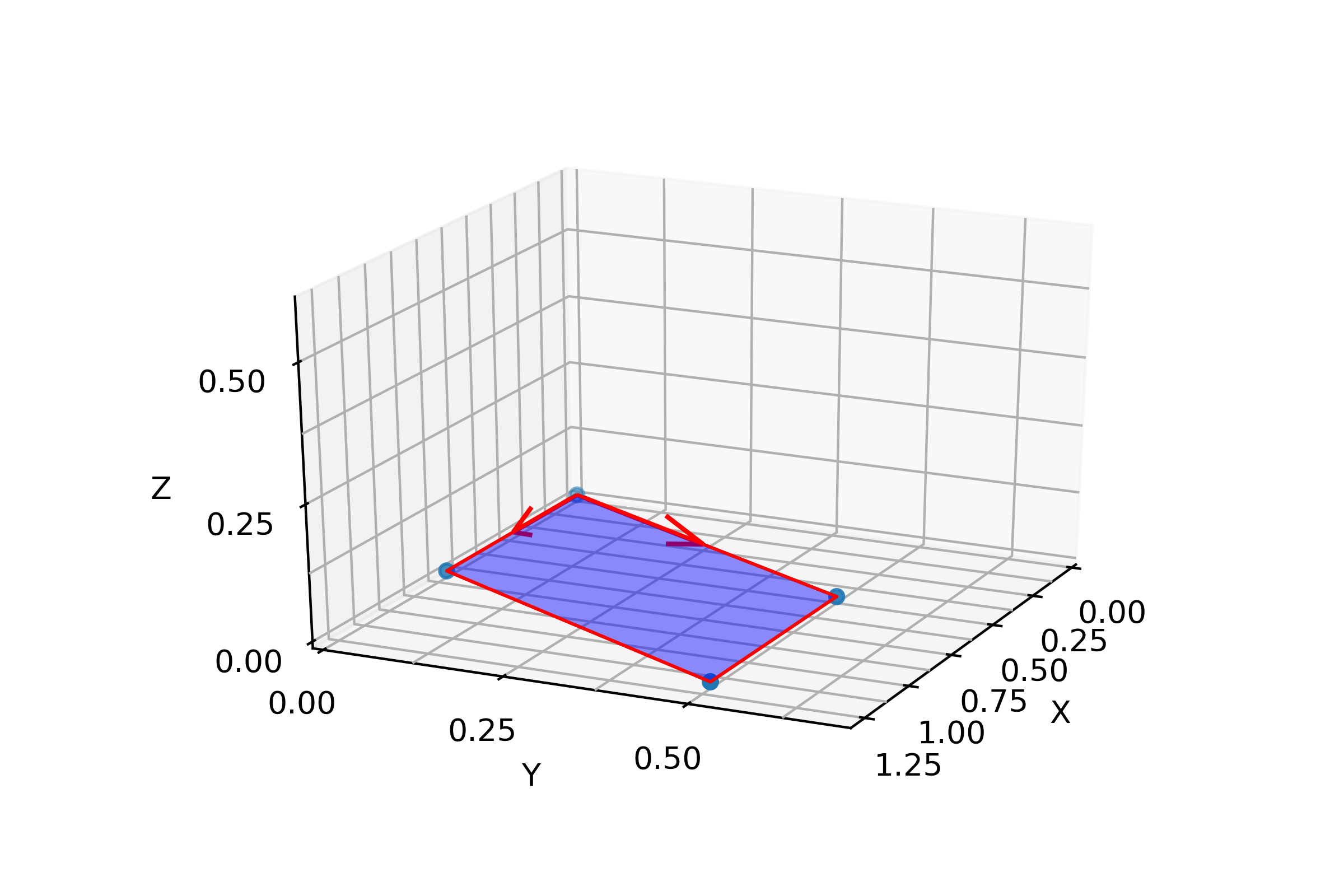}
    \end{minipage}&
    \begin{minipage}{2.5 cm}
    \textbf{Planer} \vspace{0.1 cm}\\
     No. of
     Orthogonal
     pair(Op):\\
     $Op=0$
    \end{minipage}
    &\textbf{Type I}\\
    \hline
    
    \begin{minipage}{4 cm}
    $a\ket{00}+q\ket{11}+z\ket{22}$ \vspace{0.1 cm}\\
    $~~=\ket{0}_A(a\ket{0}_B+b\ket{1}_B)$\\
    $+\ket{2}_A(z\ket{0}_B)$
    \end{minipage}& 
    \begin{minipage}{2.5 cm}
    $E_{Gm}=\frac{1}{2}$ \vspace{0.2cm}\\
    \textbf{Condition:}\\
    $|a|^2+|b|^2$\\$=|z|^2=\frac{1}{2}$\\
    \end{minipage} &
    \begin{minipage}{5 cm}
    \vspace{0.1 cm}
      \includegraphics[width=1\linewidth,trim={2.5cm 2.5cm 2.5cm 2.5cm},clip]{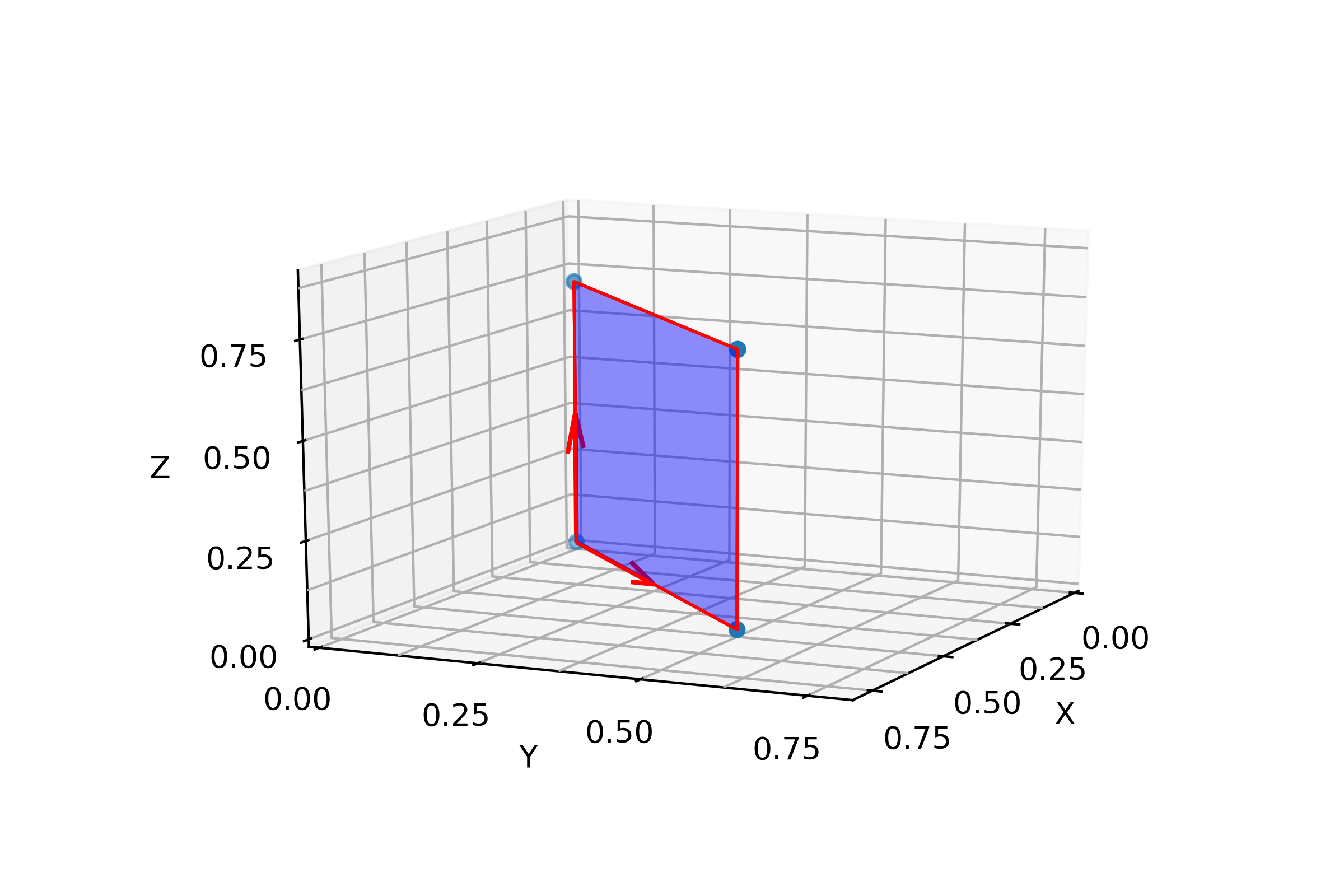}
    \end{minipage}&
    \begin{minipage}{2.5 cm}
    \textbf{Planer} \vspace{0.1 cm}\\
     No. of
     Orthogonal
     pair(Op):\\
     $Op=1$
    \end{minipage}
    &\textbf{Type I}\\
    \hline
    \end{tabular}
    \normalsize
    \caption{Geometric structure and entanglement type of three term states}
    \label{three}
\end{table}
\end{widetext}
\vspace{-0.5 cm}

\subsection{Four term states}
Four terms can be picked in 5 inequivalent ways. Example of each type of state and their properties are described in Table \ref{tab4}. Among these four term states we find a new type of entangled state (third row of Table \ref{tab4}). Geometrically this state has a non-planar shape with at least one pair of vectors non-orthogonal. This type of entangled states is named Type III entanglement.
\vspace{-0.5 cm}

\subsection{Five term states}
Five terms can be picked up in five inequivalent ways from the nine element product basis. Among all states consisting five terms, there is no new type of entanglement class. But there exists a special type of five term state (second row of Table \ref{tab5}) which can belong to any one of Type I, II, III entanglement class depending on the value of the coefficients. That state shows Type I entanglement when $\frac{|a|}{|p|}=\frac{|b|}{|q|}$ and type III for $a^*p+b^*q\neq 0$.
\vspace{-0.5 cm}
\subsection{Six Term States}
Six terms can be picked up in four inequivalent ways.Similar to five term states, no new classes are found from six term states also. But, here are some non-planer states with no mutually orthogonal post-measurement vectors (row two and three of Table \ref{six}). These states belong to Type III entangled states. 

States consisting seven basis terms or more give neither new classes nor new geometric shapes different from the ones already listed.
\begin{widetext}

\begin{table}[htbp]
    \centering
    \begin{tabular}{|c|c|c|c|c|}
    \hline
    \textbf{Example}&\textbf{Max. Ent.}&\textbf{Figure}&\textbf{Geometry}&\textbf{Ent.}\\
    \textbf{State ($\ket{\psi}$)}&\textbf{Value}&\textbf{(Example)}&\textbf{Structure}&\textbf{Class}\\
    \hline
    \begin{minipage}{4.25 cm}
    \small
    $a\ket{00}+b\ket{01}+c\ket{02}+p\ket{10}$ \vspace{0.1 cm}\\
    $=\ket{0}_A(a\ket{0}_B+b\ket{1}_B$\\
    $~~+c\ket{2}_B)+\ket{1}_A(p\ket{0}_B)$
    \end{minipage}& 
    \begin{minipage}{2.5 cm}
    \small
    $E_{Gm}<\frac{1}{2}$ \vspace{0.2cm}\\
    \textbf{Remarks:}\\
    Maximizes as $a\rightarrow 0; |b|^2+|c|^2,|p|^2\rightarrow\frac{1}{2}$
    \end{minipage} &
    \begin{minipage}{5 cm}
    \vspace{0.1 cm}
      \includegraphics[width=1\linewidth,trim={2.5cm 2.5cm 2.5cm 2.5cm},clip]{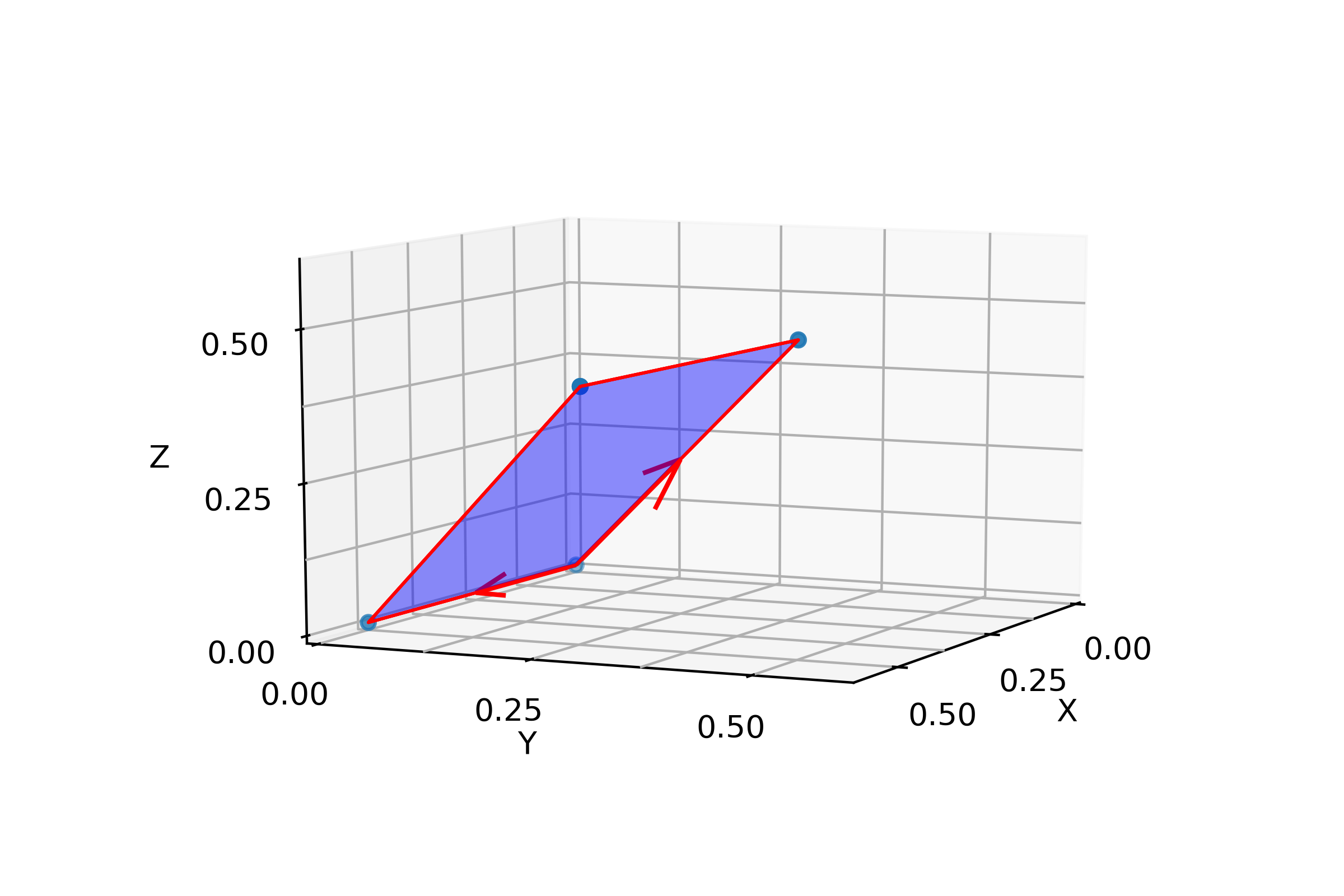}
    \end{minipage}&
    \begin{minipage}{2.5 cm}
    \small
    \textbf{Planer} \vspace{0.1 cm}\\
     No. of
     Orthogonal
     pair(Op):\\
     $Op=0$
    \end{minipage}
    &\textbf{Type I}\\
    \hline
    
    \begin{minipage}{4.25 cm}
    \small
    $a\ket{00}+b\ket{01}+p\ket{10}+r\ket{12}$ \vspace{0.1 cm}\\
    $=\ket{0}_A(a\ket{0}_B+b\ket{1}_B)$\\
    $~~+\ket{1}_A(p\ket{0}_B+r\ket{2}_B)$
    \end{minipage}& 
    \begin{minipage}{2.5 cm}
    \small
    $E_{Gm}<\frac{1}{2}$ \vspace{0.2cm}\\
    \textbf{Remarks:}\\
    Maximizes as $p \rightarrow 0 ; |a|^2+|b|^2,|r|^2 \rightarrow \frac{1}{2}$
    \end{minipage} &
    \begin{minipage}{5 cm}
    \vspace{0.1 cm}
      \includegraphics[width=1\linewidth,trim={2.5cm 2.5cm 2.5cm 2.5cm},clip]{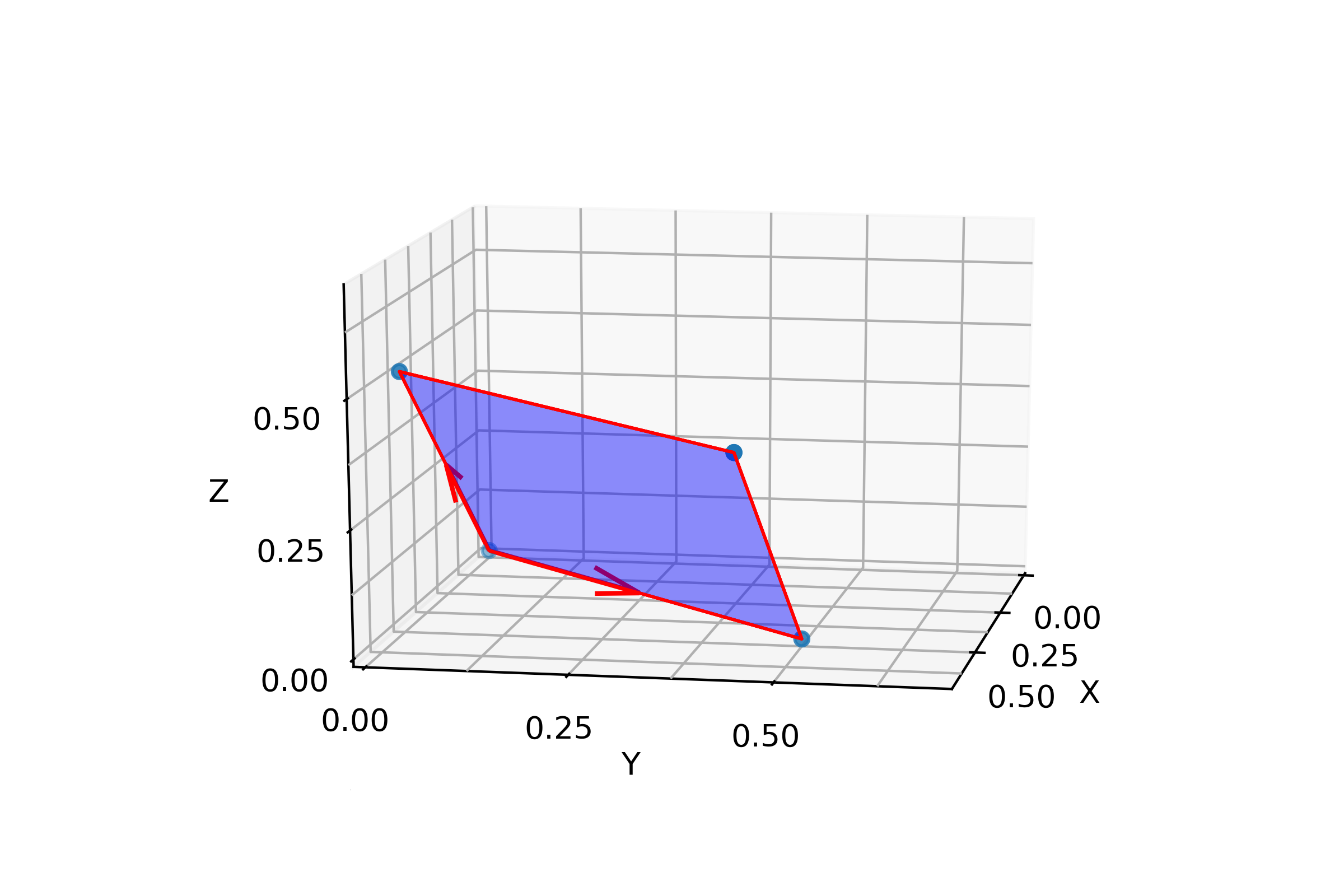}
    \end{minipage}&
    \begin{minipage}{2.5 cm}
    \small
    \textbf{Planer} \vspace{0.1 cm}\\
     No. of
     Orthogonal
     pair(Op):\\
     $Op=0$
    \end{minipage}
    &\textbf{Type I}\\
    \hline
    
    \begin{minipage}{4.25 cm}
    \small
    $a\ket{00}+b\ket{01}+p\ket{10}+z\ket{22}$ \vspace{0.1 cm}\\
    $=\ket{0}_A(a\ket{0}_B+b\ket{1}_B)+$\\$\ket{1}_A(p\ket{0}_B)+\ket{2}_A(z\ket{2}_B)$
    \end{minipage}& 
    \begin{minipage}{2.5 cm}
    \small
    $E_{Gm}<1$ \vspace{0.2cm}\\
    \textbf{Remarks:}\\
    Maximizes as $a \rightarrow 0; |b|^2,$\\$ |p|^2, |z|^2\rightarrow \frac{1}{3}$
    \end{minipage} &
    \begin{minipage}{5 cm}
    \vspace{0.1 cm}
      \includegraphics[width=1\linewidth,trim={2.5cm 2.5cm 2.5cm 2.5cm},clip]{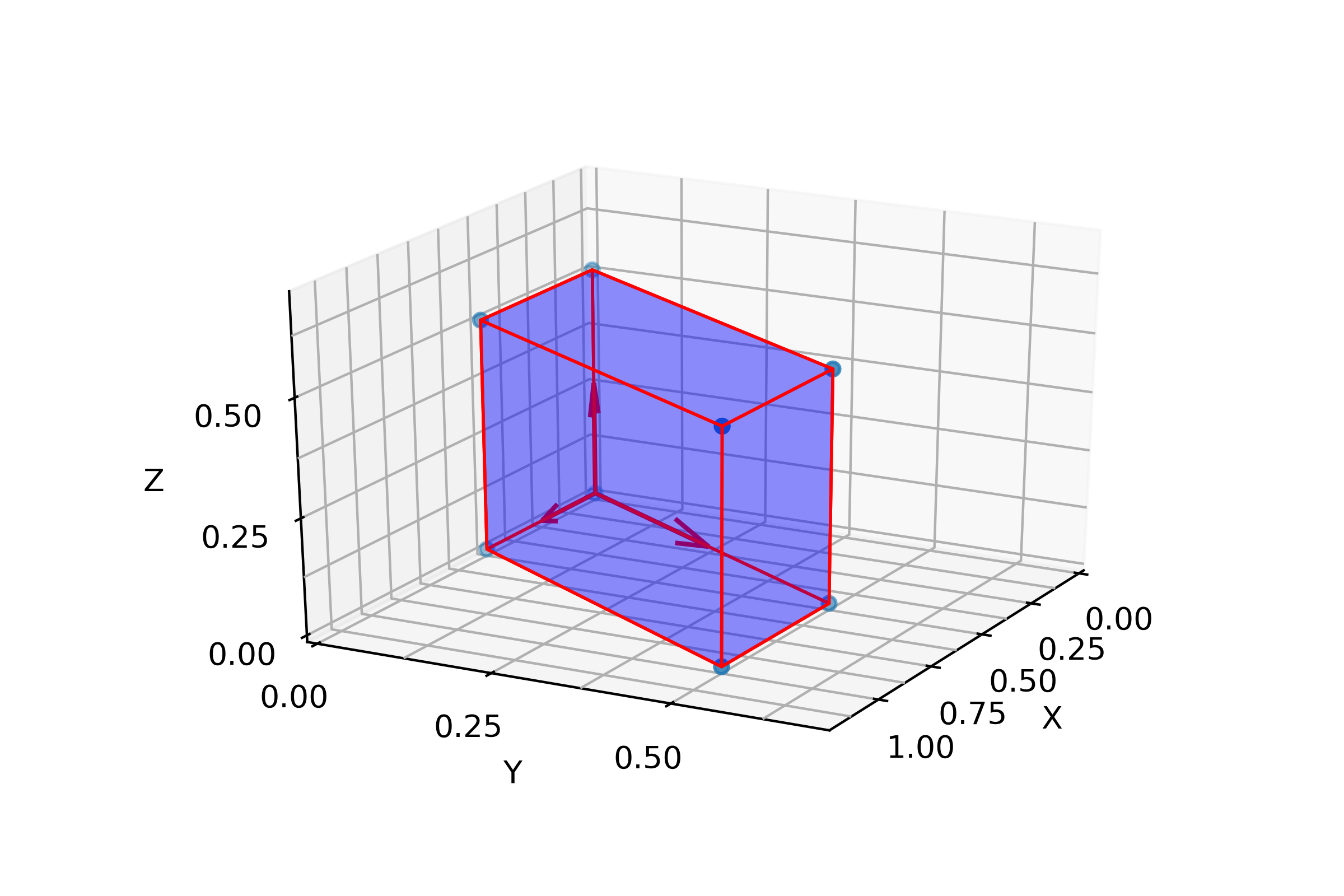}
    \end{minipage}&
    \begin{minipage}{2.5 cm}
    \small
    \textbf{Non-Planer} \vspace{0.1 cm}\\
     No. of
     Orthogonal
     pair(Op):\\
     $Op=2$
    \end{minipage}
    &\textbf{Type III}\\
    \hline
    
    \begin{minipage}{4.25 cm}
    \small
    $a\ket{00}+b\ket{01}+p\ket{10}+q\ket{11}$ \vspace{0.1 cm}\\
    $=\ket{0}_A(a\ket{0}_B+b\ket{1}_B)$\\
    $~~+\ket{1}_A(p\ket{0}_B+q\ket{1}_B)$
    \end{minipage}& 
    \begin{minipage}{2.5 cm}
    \small
    \vspace{0.1 cm}
    $E_{Gm}=\frac{1}{2}$ \vspace{0.1cm}\\
    \textbf{Condition:}\\
    $|a|^2+|b|^2=|p|^2+|q|^2=\frac{1}{2};$\\$ a^*p+b^*q=0$
    \vspace{0.1 cm}
    \end{minipage} &
    \begin{minipage}{5 cm}
    \vspace{0.1 cm}
      \includegraphics[width=1\linewidth,trim={2.5cm 2.5cm 2.5cm 2.5cm},clip]{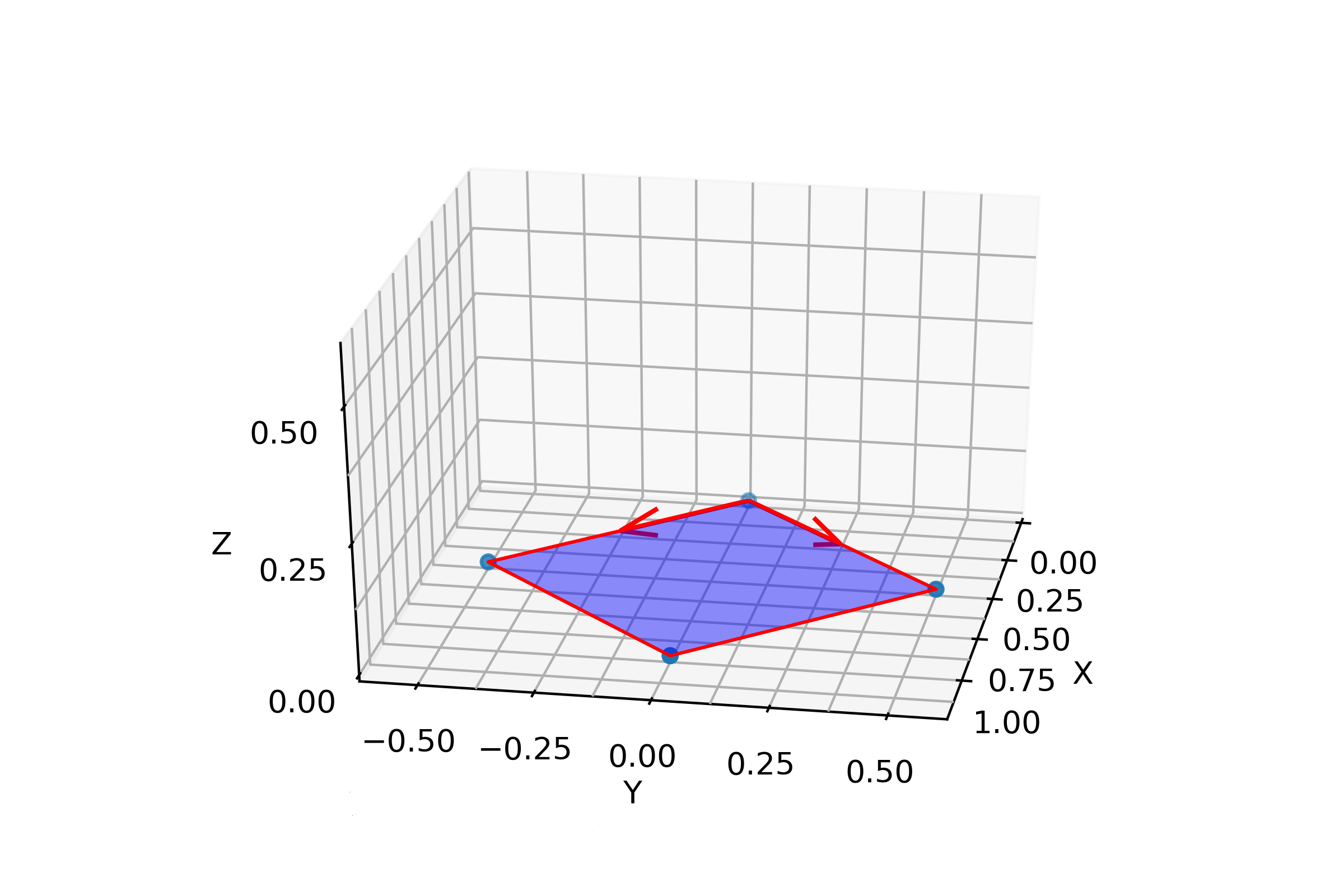}
    \end{minipage}&
    \begin{minipage}{2.5 cm}
    \small
    \textbf{Planer} \vspace{0.1 cm}\\
     No. of
     Orthogonal
     pair(Op):\\
     Op=0/1.
    \end{minipage}
    &\textbf{Type I}\\
    \hline
    
    \begin{minipage}{4.25 cm}
    \small
    $a\ket{00}+b\ket{01}+r\ket{12}+z\ket{22}$ \vspace{0.1 cm}\\
    $=\ket{0}_A(a\ket{0}_B+b\ket{1}_B)+$\\
    $\ket{1}_A(r\ket{2}_B)+\ket{2}_A(z\ket{2}_B)$
    \end{minipage}& 
    \begin{minipage}{2.5 cm}
    \small
    $E_{Gm}=\frac{1}{2}$ \vspace{0.2cm}\\
    \textbf{Condition:}\\
    $|a|^2+|b|^2=|r|^2+|z|^2=\frac{1}{2} $
    \end{minipage} &
    \begin{minipage}{5 cm}
    \vspace{0.1 cm}
      \includegraphics[width=1\linewidth,trim={2.5cm 2.5cm 2.5cm 2.5cm},clip]{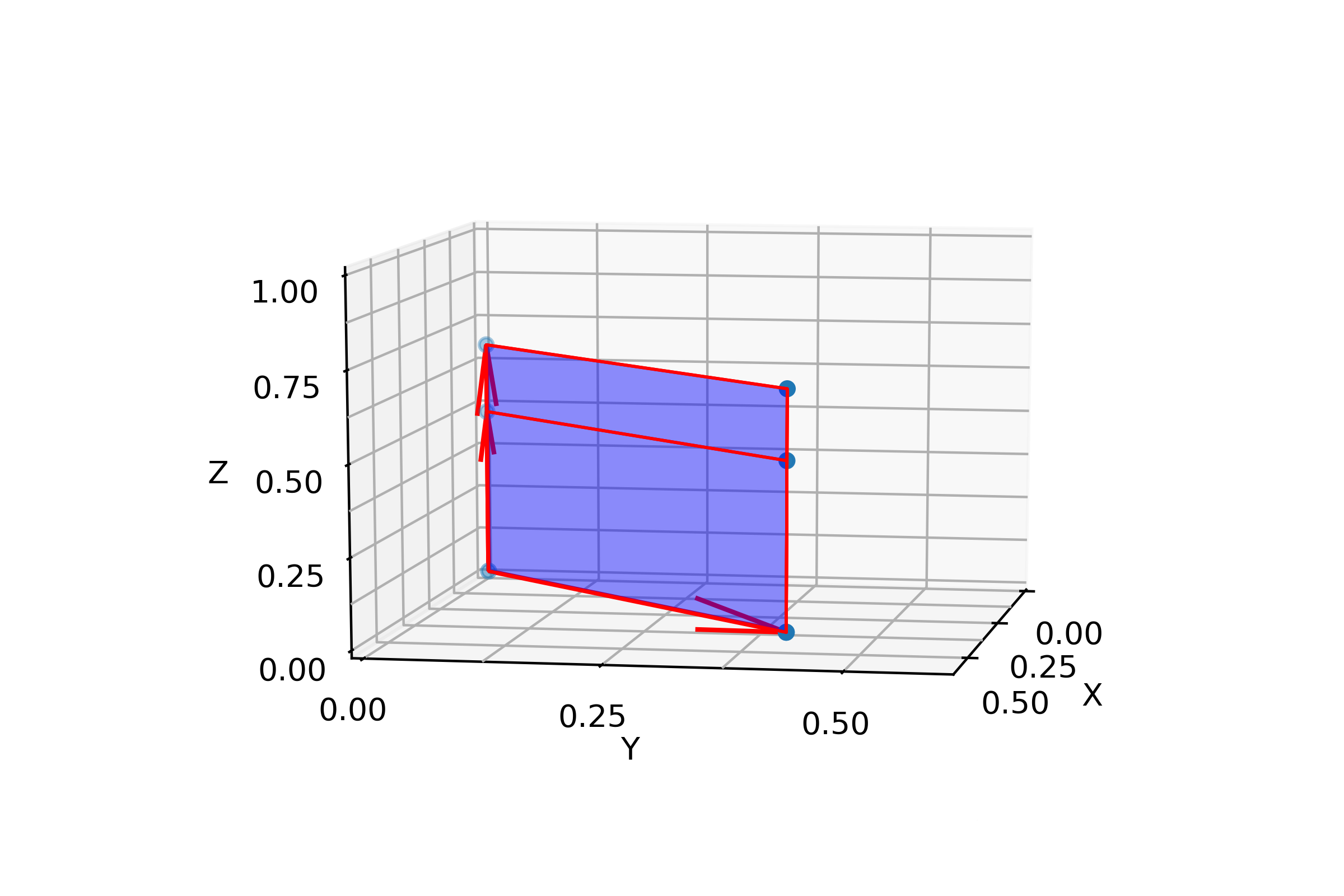}
    \end{minipage}&
    \begin{minipage}{2.5 cm}
    \small
    \textbf{Planer} \vspace{0.1 cm}\\
     No. of
     Orthogonal
     pair(Op):\\
     $Op=2$
    \end{minipage}
    &\textbf{Type I}\\
    \hline
    \end{tabular}
    \caption{Geometric structure and entanglement type of four term states}
    \label{tab4}
\end{table}

\begin{table}[htbp]
    \centering
    \begin{tabular}{|c|c|c|c|c|}
    \hline
    \textbf{Example}&\textbf{Max. Ent.}&\textbf{Figure}&\textbf{Geometry}&\textbf{Ent.}\\
    \textbf{State ($\ket{\psi}$)}&\textbf{Value}&\textbf{(Example)}&\textbf{Structure}&\textbf{Class}\\
    \hline
    \begin{minipage}{4.25 cm}
    \small
    $a\ket{00}+b\ket{01}+c\ket{02}+p\ket{10}+q\ket{11}$ \vspace{0.1 cm}\\
    $=\ket{0}_A(a\ket{0}_B+b\ket{1}_B+$\\
    $c\ket{2}_B)+\ket{1}_A(p\ket{0}_B+q\ket{1}_B)$
    \end{minipage}&
    \begin{minipage}{2.5 cm}
    \small
    $E_{Gm}=\frac{1}{2}$ \vspace{0.1cm}\\
    \textbf{Condition:}\\
    $|a|^2+|b|^2+$\\$|c|^2=|p|^2$\\$+|q|^2=\frac{1}{2},$ \\ $a^*p+b^*q=0$.
    \end{minipage} &
    \begin{minipage}{5 cm}
    \vspace{0.05 cm}
      \includegraphics[width=1\linewidth,trim={2.5cm 2.5cm 2.5cm 2.5cm},clip]{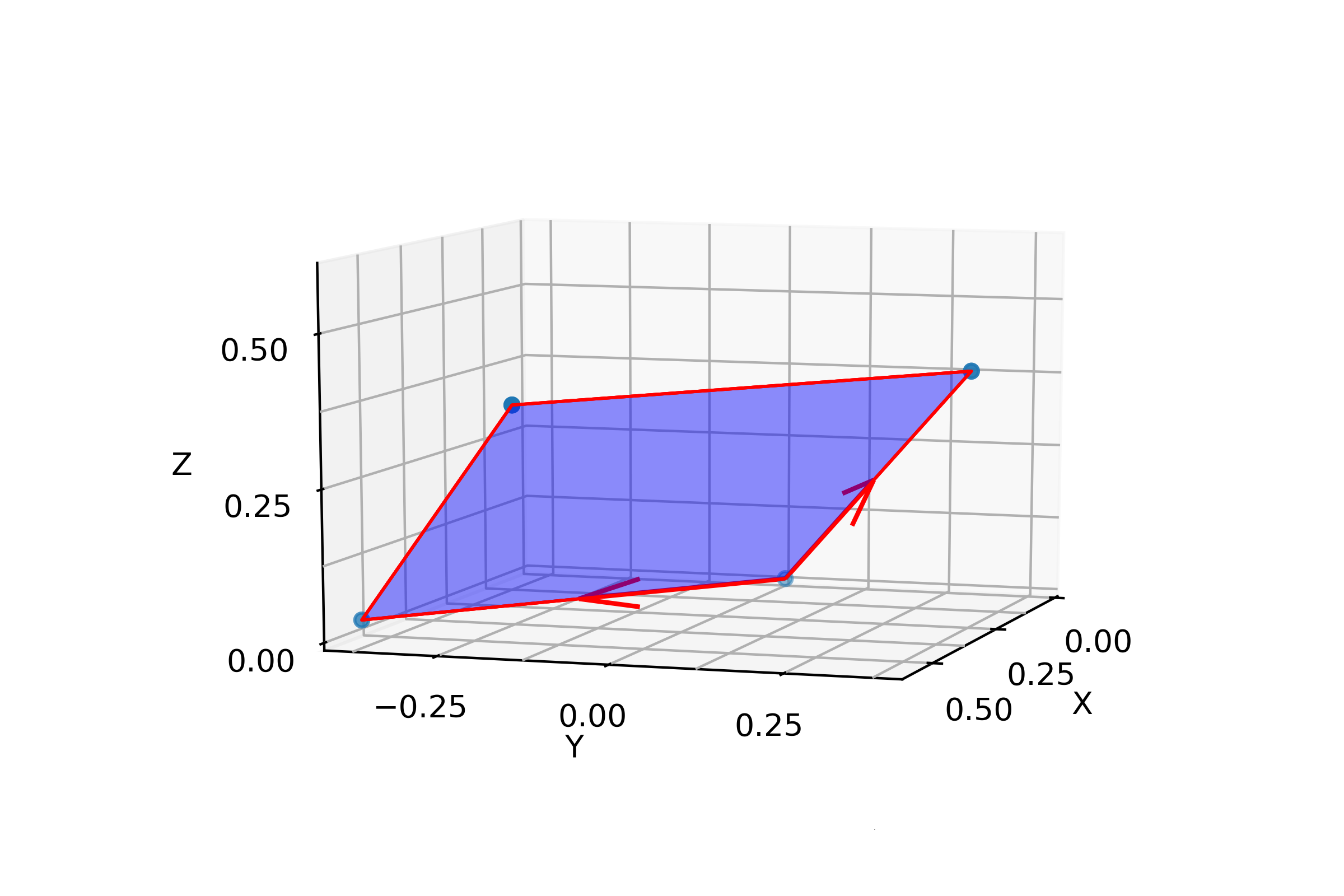}
    \end{minipage}&
    \begin{minipage}{2.5 cm}
    \small
    \textbf{Planer} \vspace{0.1 cm}\\
     No. of
     Orthogonal
     pair(Op):\\
     Op=0/1.
    \end{minipage}
    &\textbf{Type I}\\
    \hline
    
    \begin{minipage}{4.25 cm}
    \small
    $a\ket{00}+b\ket{01}+p\ket{10}+q\ket{11}+z\ket{22}$ \vspace{0.1 cm}\\
    $~~=\ket{0}_A(a\ket{0}_B+b\ket{1}_B)$\\
    $+\ket{1}_A(p\ket{0}_B+q\ket{1}_B)$\\
    $+\ket{2}_A(\ket{2}_B)$
    \end{minipage}&
    \begin{minipage}{2.5 cm}
    \small
    $E_{Gm}=1$ \vspace{0.1cm}\\
    \textbf{Condition:}\\
    $|a|^2+|b|^2=|p|^2+|q|^2=|z|^2=\frac{1}{\sqrt{3}},$\\ $a^*p+b^*q=0$.
    \end{minipage} &
    \begin{minipage}{5 cm}
    \vspace{0.05 cm}
      \includegraphics[width=1\linewidth,trim={2.5cm 2.5cm 2.5cm 2.5cm},clip]{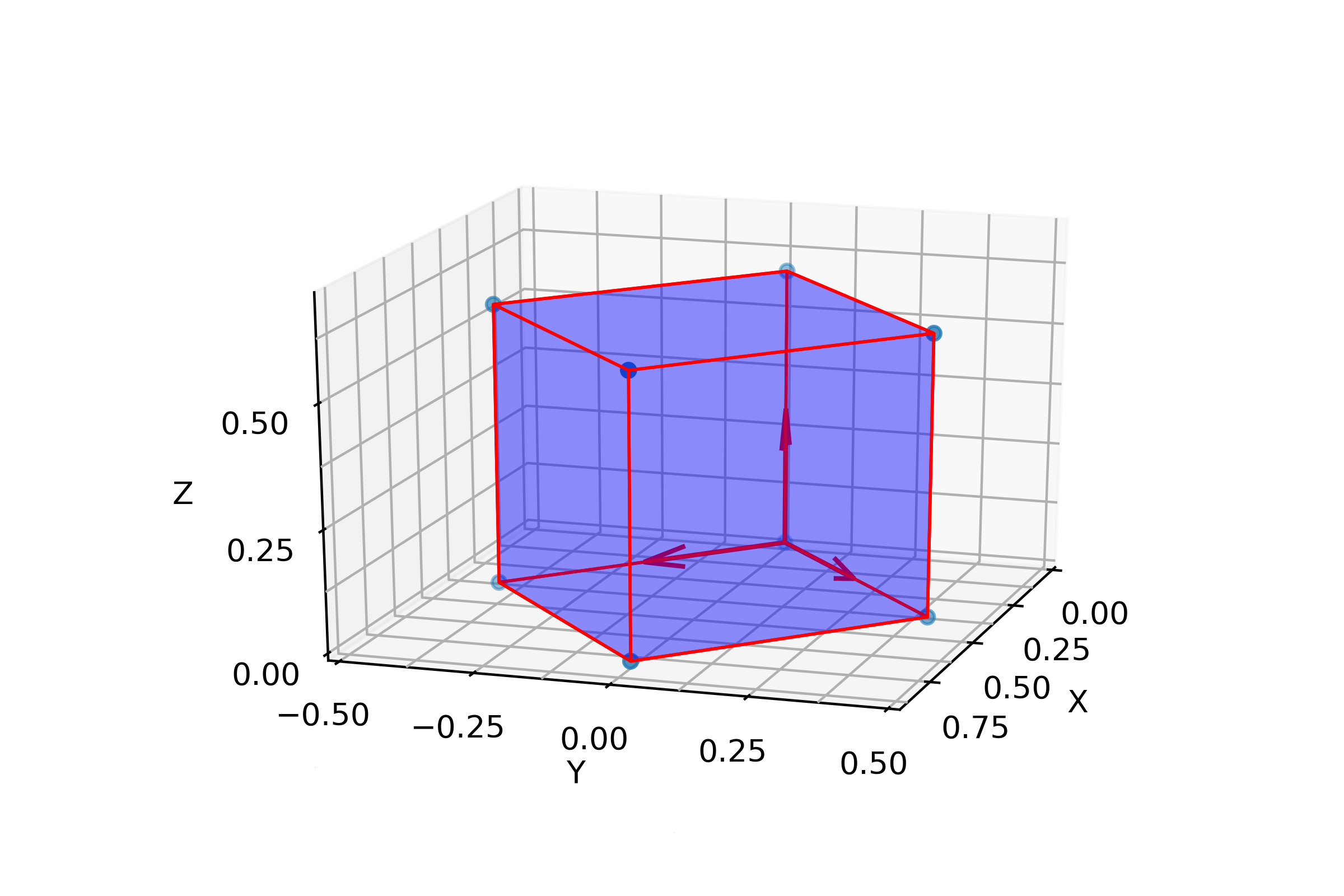}
    \end{minipage}&
    \begin{minipage}{2.5 cm}
    \small \vspace{0.2 cm}
    \textbf{Planer/ Non-Planer} \vspace{0.1 cm}\\
     No. of
     Orthogonal
     pair(Op):\\
     Op=2/3.\\
    \end{minipage}&
    \begin{minipage}{1.55 cm}
    \textbf{Type I/}\\
    \textbf{Type II/}\\
    \textbf{Type III}
    \end{minipage}\\
    \hline
    
    \begin{minipage}{4.25 cm}
    \small
    $a\ket{00}+b\ket{01}+c\ket{02}+p\ket{10}+x\ket{20}$ \vspace{0.1 cm}\\
    $~~=\ket{0}_A(a\ket{0}_B+b\ket{1}_B$\\
    $+c\ket{2}_B)+\ket{1}_A(p\ket{0}_B)$\\
    $+\ket{2}_A(x\ket{0}_B)$
    \end{minipage}&
    \begin{minipage}{2.5 cm}
    \small
    $E_{Gm}<\frac{1}{2}$ \vspace{0.1cm}\\
    \textbf{Remarks:}\\
    Maximizes as $a\rightarrow 0$ ; \\ $|b|^2+|c|^2,$\\$|p|^2+|x|^2\rightarrow\frac{1}{2}$
    \end{minipage} &
    \begin{minipage}{5 cm}
    \vspace{0.05 cm}
      \includegraphics[width=1\linewidth,trim={2.5cm 2.5cm 2.5cm 2.5cm},clip]{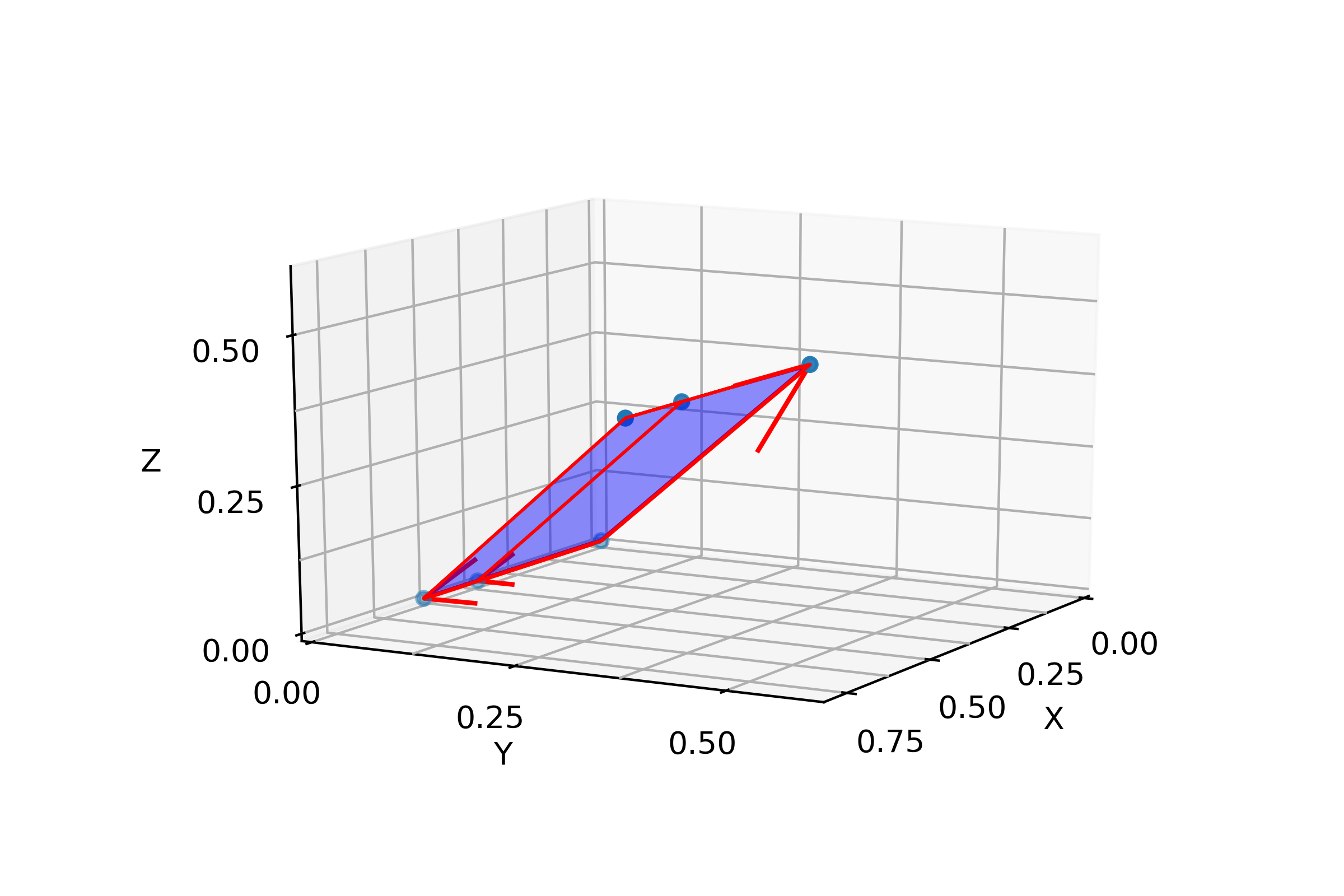}
    \end{minipage}&
    \begin{minipage}{2.5 cm}
    \small
    \textbf{Planer} \vspace{0.1 cm}\\
     No. of
     Orthogonal
     pair(Op):\\
     Op=2.
    \end{minipage}
    &\textbf{Type I}\\
    \hline
    
    \begin{minipage}{4.25 cm}
    \small
    $a\ket{00}+b\ket{01}+c\ket{02}+p\ket{10}+y\ket{21}$ \vspace{0.1 cm}\\
    $~~=\ket{0}_A(a\ket{0}_B+b\ket{1}_B$\\
    $+c\ket{2}_B)+\ket{1}_A(p\ket{0}_B)$\\
    $+\ket{2}_A(y\ket{1}_B)$
    \end{minipage}&
    \begin{minipage}{2.5 cm}
    \small
    $E_{Gm}<1$ \vspace{0.1cm}\\
    \textbf{Remarks:}\\
    Maximizes as $a,b\rightarrow 0$; \\ $ |c|^2,|p|^2,|y|^2\rightarrow \frac{1}{3}$
    \end{minipage} &
    \begin{minipage}{5 cm}
    \vspace{0.05 cm}
      \includegraphics[width=1\linewidth,trim={2.5cm 2.5cm 2.5cm 2.5cm},clip]{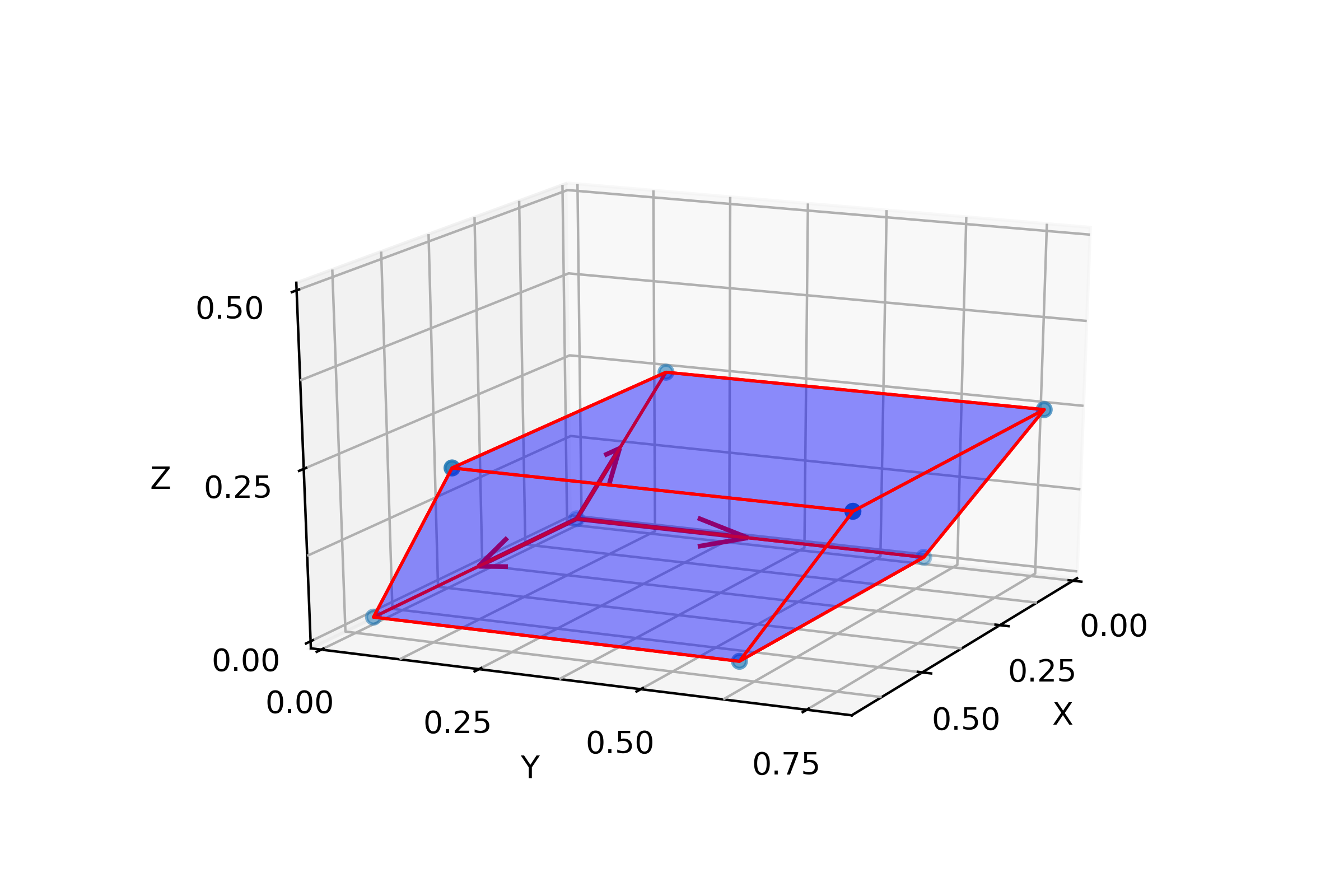}
    \end{minipage}&
    \begin{minipage}{2.5 cm}
    \small
    \textbf{Non-Planer} \vspace{0.1 cm}\\
     No. of
     Orthogonal
     pair(Op):\\
     Op=1.
    \end{minipage}
    &\textbf{Type III}\\
    \hline
    
    \begin{minipage}{4.25 cm}
    \small
    $a\ket{00}+b\ket{01}+p\ket{10}+r\ket{12}+y\ket{21}$ \vspace{0.1 cm}\\
    $~~=\ket{0}_A(a\ket{0}_B+b\ket{1}_B)$\\
    $+\ket{1}_A(p\ket{0}_B+r\ket{2}_B)$\\
    $+\ket{2}_A(y\ket{1}_B)$
    \end{minipage}&
    \begin{minipage}{2.5 cm}
    \small
    $E_{Gm}<1$ \vspace{0.1cm}\\
    \textbf{Remarks:}\\
    Maximizes as $b,p\rightarrow 0$; \\ $ |a|^2,|r|^2,|y|^2\rightarrow \frac{1}{3}$
    \end{minipage} &
    \begin{minipage}{5 cm}
    \vspace{0.05 cm}
      \includegraphics[width=1\linewidth,trim={2.5cm 2.5cm 2.5cm 2.5cm},clip]{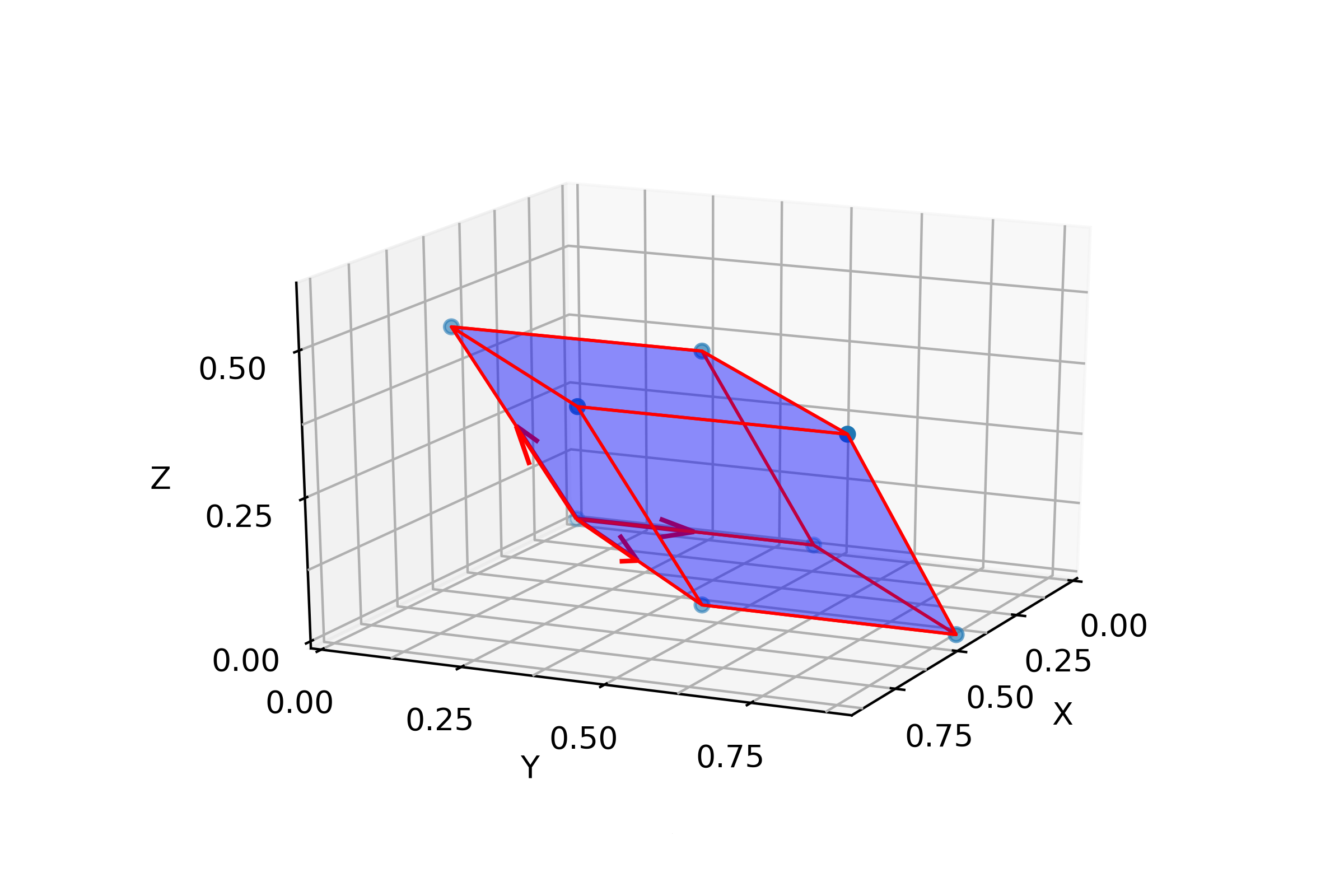}
    \end{minipage}&
    \begin{minipage}{2.5 cm}
    \small
    \textbf{Non-Planer} \vspace{0.1 cm}\\
     No. of
     Orthogonal
     pair(Op):\\
     Op=1.
    \end{minipage}
    &\textbf{Type III}\\
    \hline
    
    \begin{minipage}{4.25 cm}
    \small
    $a\ket{00}+b\ket{01}+p\ket{10}+r\ket{12}+z\ket{22}$ \vspace{0.1 cm}\\
    $~~=\ket{0}_A(a\ket{0}_B+b\ket{1}_B)$\\
    $+\ket{1}_A(p\ket{0}_B+r\ket{2}_B)$\\
    $+\ket{2}_A(z\ket{2}_B)$
    \end{minipage}&
    \begin{minipage}{2.5 cm}
    \small
    $E_{Gm}<1$ \vspace{0.1cm}\\
    \textbf{Remarks:}\\
    Maximizes as $a,r\rightarrow 0$; \\ $ |b|^2,|p|^2,|z|^2\rightarrow \frac{1}{3}$
    \end{minipage} &
    \begin{minipage}{5 cm}
    \vspace{0.05 cm}
      \includegraphics[width=1\linewidth,trim={2.5cm 2.5cm 2.5cm 2.5cm},clip]{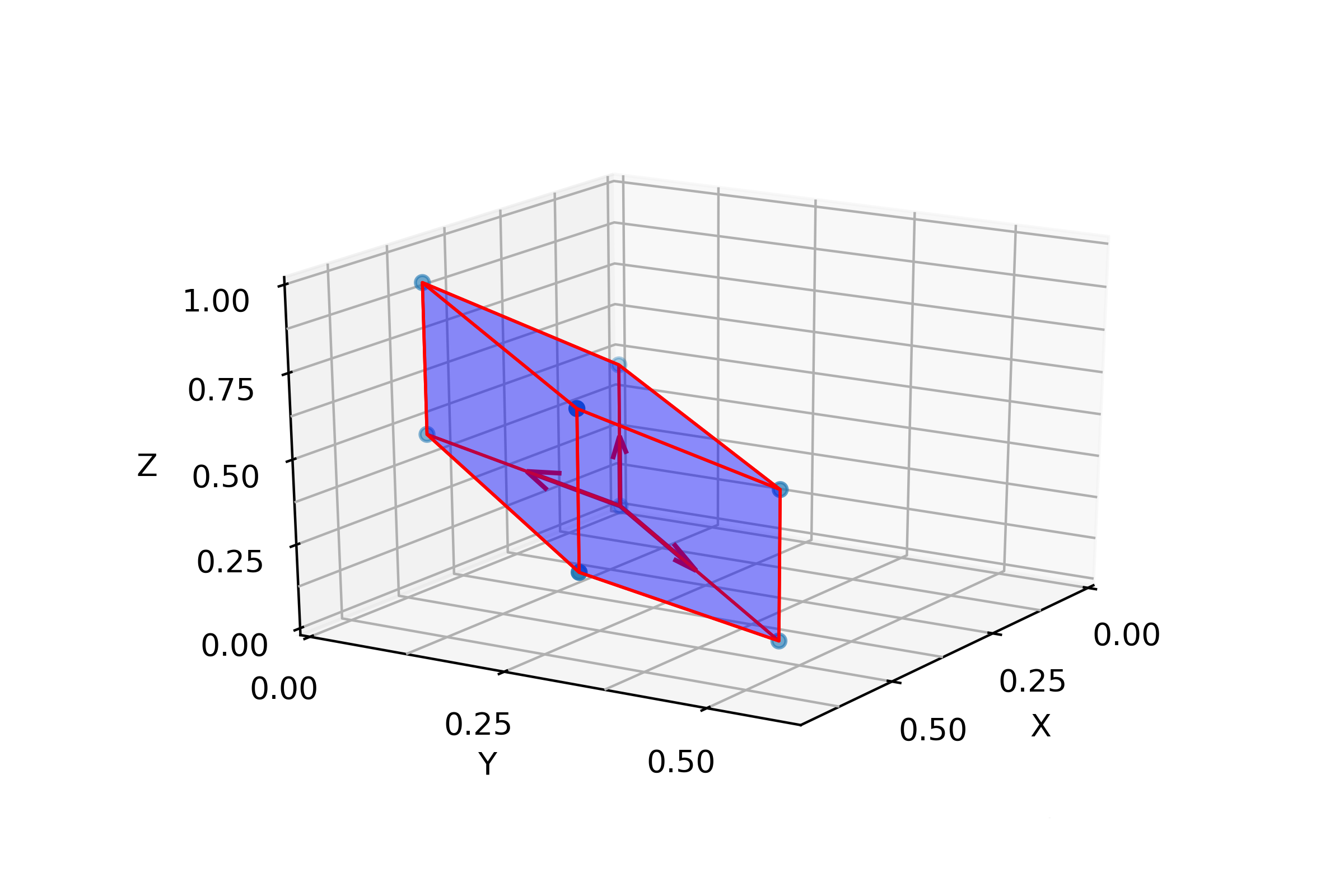}
    \end{minipage}&
    \begin{minipage}{2.5 cm}
    \small
    \textbf{Non-Planer} \vspace{0.1 cm}\\
     No. of
     Orthogonal
     pair(Op):\\
     Op=1.
    \end{minipage}
    &\textbf{Type III}\\
    \hline
    \end{tabular}
    \caption{Geometric structure and entanglement type of five term states}
    \label{tab5}
\end{table}
\end{widetext}

\section{Generalization to Higher Dimension}\label{section 5}
In this section we extend our entanglement measure for multi qudit systems. First, we present the entanglement measure for two qudit systems. Any two qudit system of dimension $d_1$ and $d_2$ respectively can be represented as $\ket{\psi}=\sum_{i,j=0}^{d_1-1,d_2-1}a_{ij}\ket{i}_A\ket{j}_B$ where $a_{ij}$ are complex coefficients and i and j are corresponding basis elements. In terms of post measurement vectors it can be rewritten as $\ket{\psi}=\sum_{i=0}^{d_1-1}\ket{i}_A\ket{\gamma_i}_B$ where $\ket{\gamma_i}=\sum_{j=0}^{d_2-1}a_{ij}\ket{j}$ are the post-measurement vectors. Hence, the entanglement measure for these bipartite qudit is given by:
\begin{equation}
    \small
        E_G=|\wedge_{i=0}^{d_1-1}\ket{\gamma_i}^2+\sum_{i=0}^{d_1-1}\wedge_{j\neq i}^{d_1-1}\ket{\gamma_j}^2+\cdots+\frac{1}{2}\sum_{i,j\neq 0}^{d_1-1}|\ket{\gamma_i}\wedge\ket{\gamma_j}|^2
\end{equation}
where we define, $\wedge_{k=0}^p\ket{k}=\ket{0}\wedge\ket{1}\wedge\cdots\wedge\ket{p}$ and $\wedge_{k\neq l}^p\ket{k}=\ket{0}\wedge \cdots \wedge \ket{l-1}\wedge\ket{l+1}\wedge\cdots\ket{p}$. 

Consider any n-qudit pure state $\ket{\psi}$ of arbitrary dimension. It can be represented as:
\begin{equation}
    \ket{\psi}=\sum_{i_1,i_2,\cdots,i_n=0}^{d_1-1,d_2-1,\cdots ,d_n-1}a_{i_1\cdots i_n}\ket{i_1}\cdots \ket{i_n}
\end{equation}
Here, $\ket{i_i}$ are the basis elements for i-th qudit of dimension $d_i$. Consider any bipartition of $(m|n-m)$ and we look into the post-measurement vectors in the $(n-m)$ subsystem. 

\begin{widetext}

\begin{table}[htbp]\label{six}
\centering
    \begin{tabular}{|c|c|c|c|c|}
    \hline
    \textbf{Example}&\textbf{Max. Ent.}&\textbf{Figure}&\textbf{Geometry}&\textbf{Ent.}\\
    \textbf{State ($\ket{\psi}$)}&\textbf{Value}&\textbf{(Example)}&\textbf{Structure}&\textbf{Class}\\
    \hline
    \begin{minipage}{4 cm}
    \small
    $a\ket{00}+b\ket{01}+c\ket{02}+p\ket{10}+q\ket{11}+r\ket{12}$ \vspace{0.1 cm}\\
    $~~=\ket{0}_A(a\ket{0}_B+b\ket{1}_B$\\
    $+c\ket{2}_B)+\ket{1}_A(p\ket{0}_B$\\
    $+q\ket{1}_B+r\ket{2}_B)$
    \end{minipage}&
    \begin{minipage}{2.5 cm}
    \small
    \vspace{0.2 cm}
    $E_{Gm}=\frac{1}{2}$ \vspace{0.1cm}\\
    \textbf{Condition:}\\
    $|a|^2+|b|^2+$\\$|c|^2=|p|^2$\\$+|q|^2+|r|^2$\\$=\frac{1}{2};$ $a^*p+$\\$b^*q+c^*r=0$.\vspace{0.2 cm}
    \end{minipage} &
    \begin{minipage}{5 cm}
    \vspace{0.1 cm}
      \includegraphics[width=1\linewidth,trim={2.5cm 2.5cm 2.5cm 2.5cm},clip]{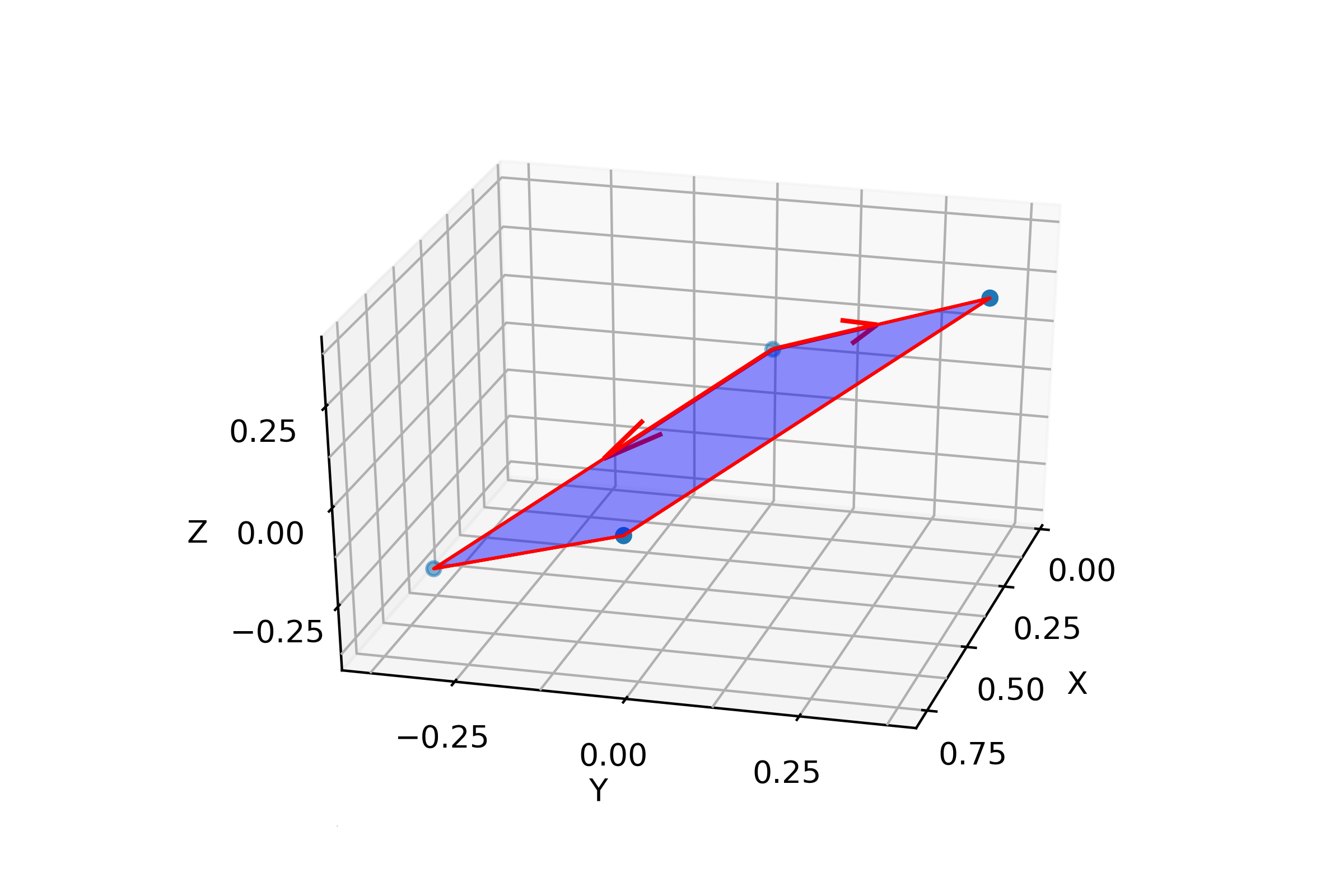}
    \end{minipage}&
    \begin{minipage}{2.5 cm}
    \small
    \textbf{Planer} \vspace{0.1 cm}\\
     No. of
     Orthogonal
     pair(Op):\\
     Op=0/1.
    \end{minipage}
    &\textbf{Type I}\\
    \hline
    
    \begin{minipage}{4 cm}
    \small
    $a\ket{00}+b\ket{01}+p\ket{10}+r\ket{12}+y\ket{21}+z\ket{22}$ \vspace{0.1 cm}\\
    $~~=\ket{0}_A(a\ket{0}_B+b\ket{1}_B)$\\
    $+\ket{1}_A(p\ket{0}_B+r\ket{2}_B)$\\
    $+\ket{2}_A(y\ket{1}_B+z\ket{2}_B)$
    \end{minipage}&
    \begin{minipage}{2.5 cm}
    \small
    $E_{Gm}<1$ \vspace{0.1cm}\\
    \textbf{Remarks:}\\
    Maximizes as $b,p,x\rightarrow $\\$ 0$; $ |a|^2,|r|^2,$\\$|y|^2\rightarrow \frac{1}{3}$
    \end{minipage} &
    \begin{minipage}{5 cm}
    \vspace{0.1 cm}
      \includegraphics[width=1\linewidth,trim={2.5cm 2.5cm 2.5cm 2.5cm},clip]{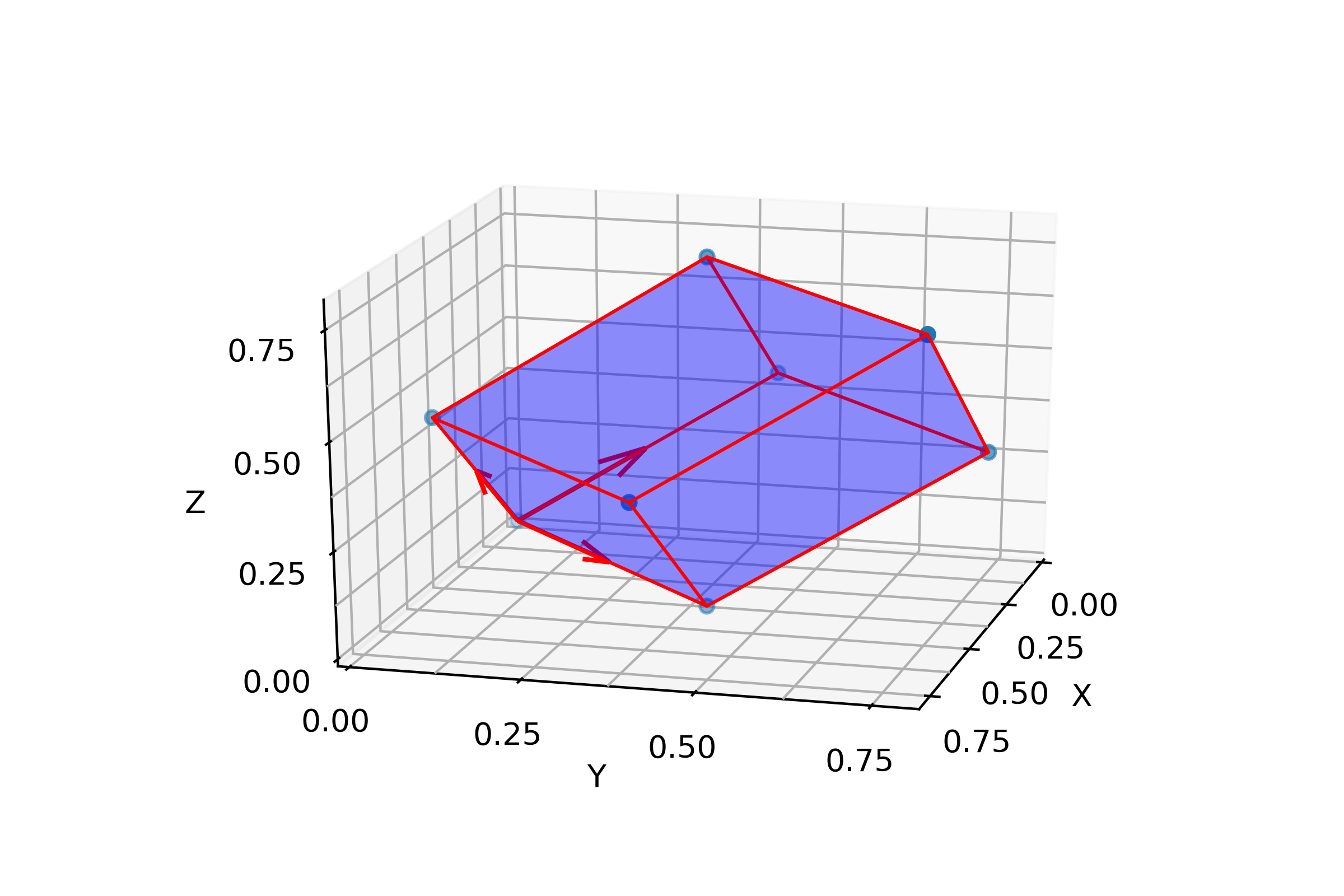}
    \end{minipage}&
    \begin{minipage}{2.5 cm}
    \small
    \textbf{Non-Planer} \vspace{0.1 cm}\\
     No. of
     Orthogonal
     pair(Op):\\
     Op=0.
    \end{minipage}
    &\textbf{Type III}\\
    \hline
    
    \begin{minipage}{4 cm}
    \small
    $a\ket{00}+b\ket{01}+c\ket{02}+p\ket{10}+q\ket{11}+x\ket{20}$ \vspace{0.1 cm}\\
    $~~=\ket{0}_A(a\ket{0}_B+b\ket{1}_B$\\
    $+c\ket{2}_B)+\ket{1}_A(p\ket{0}_B$\\
    $+q\ket{1}_B)+\ket{2}_A(x\ket{0}_B)$
    \end{minipage}&
    \begin{minipage}{2.5 cm}
    \small
    $E_{Gm}<1$ \vspace{0.1cm}\\
    \textbf{Remarks:}\\
    Maximizes as $a,b,p\rightarrow $\\$ 0$; $ |c|^2,|q|^2,$\\$|x|^2\rightarrow \frac{1}{3}$
    \end{minipage} &
    \begin{minipage}{5 cm}
    \vspace{0.1 cm}
      \includegraphics[width=1\linewidth,trim={2.5cm 2.5cm 2.5cm 2.5cm},clip]{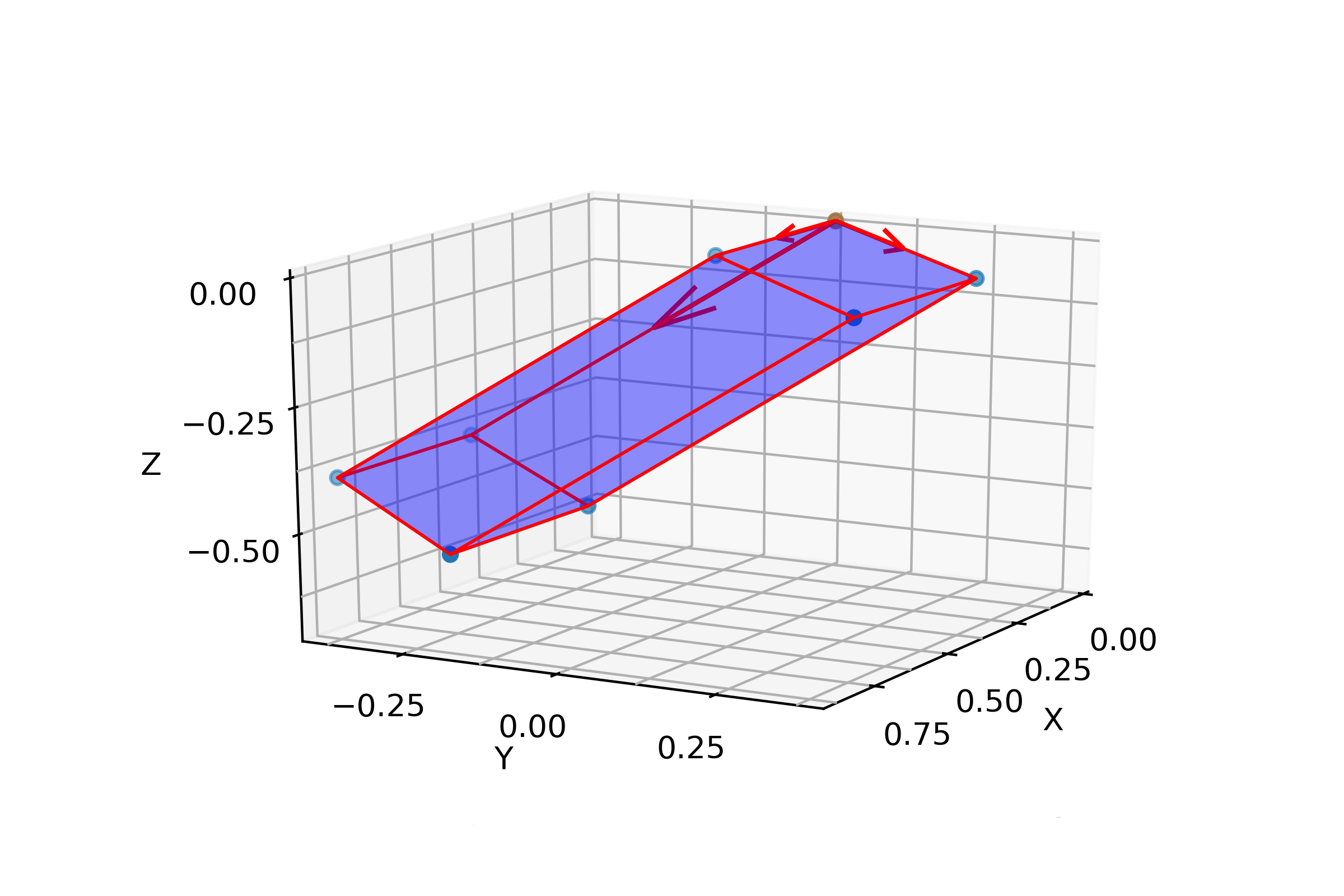}
    \end{minipage}&
    \begin{minipage}{2.5 cm}
    \small
    \textbf{Non-Planer} \vspace{0.1 cm}\\
     No. of
     Orthogonal
     pair(Op):\\
     Op=0/1.
    \end{minipage}
    &\textbf{Type III}\\
    \hline
    
    \begin{minipage}{4 cm}
    \small
    $a\ket{00}+b\ket{01}+c\ket{02}+p\ket{10}+q\ket{11}+z\ket{22}$ \vspace{0.1 cm}\\
    $~~=\ket{0}_A(a\ket{0}_B+b\ket{1}_B$\\
    $+c\ket{2}_B)+\ket{1}_A(p\ket{0}_B$\\
    $+q\ket{1}_B)+\ket{2}_A(x\ket{2}_B)$
    \end{minipage}&
    \begin{minipage}{2.5 cm}
    \small
    \vspace{0.2 cm}
    $E_{Gm}<1$ \vspace{0.1cm}\\
    \textbf{Remarks:}\\
    Maximizes as $c\rightarrow 0$;\\ $ |b|^2+|c|^2,$\\$|p|^2+|x|^2,$\\$|x|^2\rightarrow\frac{1}{3}$;\\
    $a^*p+b^*q=0$.
    \vspace{0.2 cm}
    \end{minipage} &
    \begin{minipage}{5 cm}
    \vspace{0.1 cm}
      \includegraphics[width=1\linewidth,trim={2.5cm 2.5cm 2.5cm 2.5cm},clip]{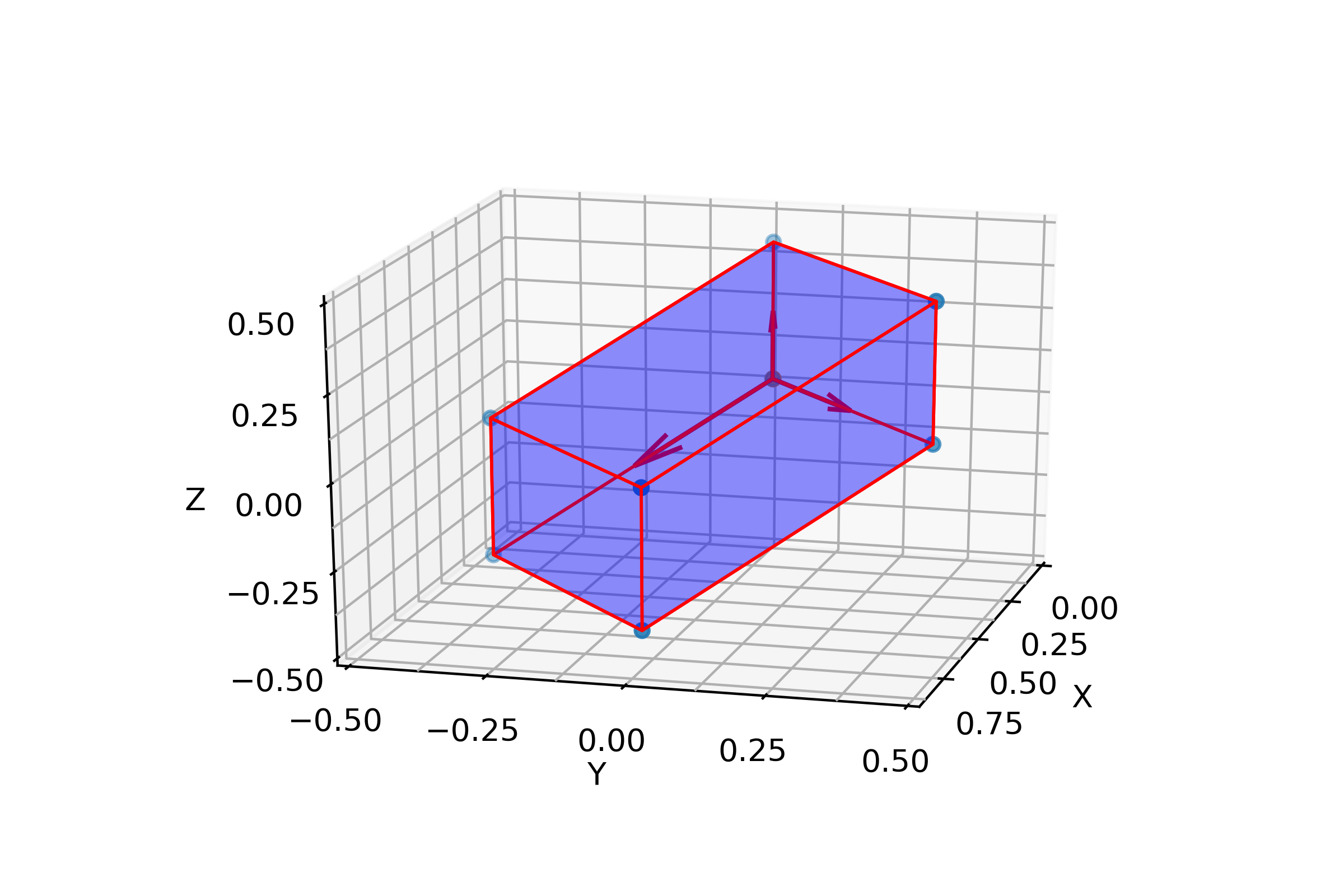}
    \end{minipage}&
    \begin{minipage}{2.5 cm}
    \small
    \textbf{Non-Planer} \vspace{0.1 cm}\\
     No. of
     Orthogonal
     pair(Op):\\
     Op=1/2.
    \end{minipage}
    &\textbf{Type III}\\
    \hline
    \end{tabular}
    \caption{Geometric structure and entanglement type of six term states}
    \end{table}
\end{widetext}

Without loss of generality we relabel those m parties as $\{1,2,\cdots,m\}$ and the remaining parties as $\{m+1,m+2,\cdots,n\}$. In terms of the post measurement vector of those n-m parties, the state can be rewritten as $\ket{\psi}=\sum_{i_1,\cdots,i_m}\ket{i_1i_2\cdots i_m}\ket{\alpha_{i_1\cdots i_m}}$ where, $\ket{\alpha_{i_1\cdots i_m}}=\sum_{i_{m+1}\cdots i_n}a_{i_1\cdots i_n}\ket{i_{m+1}i_{m+2}\cdots i_n}$ is the post-measurement vector of dimension $D=d_{m+1}.d_{m+2}\cdots d_n$. We propose the entanglement measurement for such m party bipartition as:
\begin{equation}
    \begin{split}
        E_G^m&=|\wedge_{i_1,\cdots,i_n=0}^{d_1-1,\cdots,d_n-1}\ket{\alpha_{i_1\cdots i_m}}|^2\\&+\sum_{i_k=0}^{d_k-1}|\wedge_{i_1,\cdots,i_{k-1},i_{k+1},\cdots,i_m=0}^{d_1-1,\cdots,d_{k-1}-1,d_{k+1}-1,\cdots,d_m-1}\ket{\alpha_{i_1\cdots i_{k-1}i_{k+1} \cdots i_m}}|^2+ \\
        &\cdots+\frac{1}{2}\sum_{i_1,\cdots,i_m=0}^{d_1-1,\cdots ,d_m-1}\sum_{j_1,\cdots,j_m=0}^{d_1-1,\cdots ,d_m-1}|\ket{\alpha_{i_1\cdots i_m}}\wedge\ket{\alpha_{j_1\cdots j_m}}|^2
    \end{split}
\end{equation}
So, we define the total entanglement of the n-partite qudit state as:
\begin{equation}
    E_G=\sum_{m=1}^{n-1}\mathcal{C}_m^n E_G^m
\end{equation}
where $\mathcal{C}_m^n$ denotes all possible combination of m items selected from n items.

\section{Conclusion} Our entanglement measure for bipartite qutrit systems incorporates area and volume elements of the 3 dimensional representation of the complex parallelepiped constructed by the post-measurement vectors. This geometry entirely encapsulates the distinct structures of entangled classes of those states. We have shown that there are three non-transformable types of geometries, namely three planar vectors, three mutually orthogonal vectors and three vectors neither mutually orthogonal nor planar. States with different types of geometries cannot be transformed into one another with nonzero probability under LOCC, hence they correspond to different classes of entangled states. The geometric shapes along with the maximizing condition of entanglement have been presented for all types of bipartite pure qutrit states. The maximum entanglement can be found for the type II entangled class, with the maximum value being 1. The planar structure, being type I entangled states have entanglement no more than $\frac{1}{2}$. Type III states have the maximum amount of entanglement strictly less than unity. The entanglement maximization condition from geometry has been presented for each states, which admits simplified algebric conditions. Finally, the measure has been generalized to two party arbitrary pure states and to multiparty scenario. Our results show that entanglement has an inherent connection with geometry and our entanglement measure based on wedge product formalism geometrically classifies different entangled classes as well as simplifies the entanglement maximization criteria for pure states.  

\begin{acknowledgements}
We would like to thank Mr. Abhinash Kumar Roy for numerous enlightening discussions. SM and SD would like to thank PKP and IISER Kolkata for the hospitality. PKP acknowledges the support from DST, India through Grant No. DST/ICPS/ QuST/Theme-1/2019/2020- 21/01.
\end{acknowledgements}

%
%
\bibliography{sample} 


%
%


\end{document}